\pdfoutput=1
\newif\ifFull
\Fullfalse

\documentclass{cccg13}

\usepackage{cprotect}

\usepackage{amsmath}
\usepackage{amssymb}
\usepackage{times}
\usepackage{compress}
\usepackage{algorithmic}
\usepackage[boxed]{algorithm}
\usepackage{tabularx}
\usepackage{graphicx}
\usepackage{natbib}
\usepackage{tikz}
\usetikzlibrary{%
  arrows,
  calc,
  chains,
  intersections,
  scopes,
  snakes
}

\begin{document}

\title{Cole's Parametric Search Technique Made Practical}

\author{Michael T. Goodrich \\[5pt]
Dept.~of Computer Science \\
University of California, Irvine
\and Pawe\l{} Pszona \\[5pt]
Dept.~of Computer Science \\
University of California, Irvine
}
\date{}

\maketitle

\begin{abstract}
Parametric search has been widely used in
geometric algorithms.
Cole's improvement
provides a way of saving a logarithmic factor in
the running time over what is achievable using the
standard method.
Unfortunately, this improvement comes at the expense of making an already
complicated algorithm even more complex; hence, this technique
has been mostly of theoretical interest.
In this paper,
we provide an algorithm engineering framework that allows
for the same asymptotic complexity to be achieved
probabilistically in a way that is both simple and practical
(i.e., suitable for actual implementation).
The main idea of our approach is to show that a variant
of quicksort, known as \emph{boxsort}, can be used
to drive comparisons, instead of using a sorting network,
like the complicated AKS network, or an EREW parallel
sorting algorithm, like the fairly intricate parallel mergesort algorithm.
This results in a randomized optimization algorithm with
a running time matching that of using Cole's method, with high probability,
while also being practical.
We show how this results in
practical implementations of some geometric algorithms
utilizing parametric searching and provide experimental
results that prove practicality of the method.
\end{abstract}

\section{Introduction}\label{sec_intro}

\emph{Parametric search}~\cite{DBLP:journals/jacm/Megiddo83}
has proven to be
a useful technique in design of efficient
algorithms for many geometric and combinatorial 
optimization problems (e.g., see~\cite{Agarwal:1998,Agarwal:1994,Salowe:2004}).
Example applications include 
ray shooting~\cite{agarwal:794},
slope selection~\cite{cole:792},
computing the Fr{\'e}chet distance
between two polygonal curves~\cite{DBLP:journals/ijcga/AltG95, Chambers:2008},
matching drawings of planar graphs~\cite{Alt2003262},
labeling planar maps with rectangles~\cite{knntw-lprvs-02},
and various other matching and approximation problems
(e.g., see~\cite{Duncan:1997,Fournier:2011,Goodrich:1995}).

Although it has been superseded in some applications by
Chan's randomized optimization technique~\cite{DBLP:journals/dcg/Chan99,DBLP:conf/soda/Chan04},
for many problems (most notably, involving Fr{\'e}chet distance)
asymptotically best known results still depend on parametric searching.

The technique is applied to a decision problem, $B$,
whose solution depends on a real parameter, $\lambda$, in a monotonic way, 
so that $B$ is true on some interval $(-\infty,\lambda^*)$.
The goal is to determine the value of $\lambda^*$, the maximum for
which $B$ is true.
To achieve this goal, the parametric search approach utilizes two algorithms.
The first algorithm, $\cal C$, is a sequential \emph{decision} algorithm 
for $B$ that 
can determine if a given $\lambda$ is less than, equal to, or greater than
$\lambda^*$.
The second algorithm, $\cal A$, is a \emph{generic} parallel algorithm 
whose
inner workings are driven by ``comparisons,'' which are either independent of
$\lambda$ or
depend on the signs of low-degree polynomials in $\lambda$.
Because $\cal A$ works in parallel, its comparisons come in batches,
so there are several independent such comparisons that occur at the same
time.
The idea, then, is to run $\cal A$ on the input that depends on the unknown
value $\lambda^*$, which will result in actually finding that value
as a kind of by-product (even though we do not know $\lambda^*$, $\cal C$
can be used to resolve comparisons that appear during the execution of $\cal A$).
The next step is to simulate an execution of $\cal A$ sequentially.
To resolve comparisons that occur in a single step of this simulation, 
we can use the algorithm $\cal C$
to perform binary search among the (ordered) roots
of the polynomials in $\lambda$ for these comparisons, which allows us to
determine signs of all these polynomials, hence, allows us to continue
the simulation.
When the simulation completes, we will have determined the value of
$\lambda^*$.
Moreover, the running time for performing this simulation is 
$O(P(n)T(n) + C(n) T(n)\log P(n))$, where
$C(n)$ is the (sequential) running time of $\cal C$,
$T(n)$ is the (parallel) running time of $\cal A$,
and
$P(n)$ is the number of processors used by $\cal A$.

Cole~\cite{DBLP:journals/jacm/Cole87} shows how to improve the asymptotic
performance of the parametric search technique when sorting is
the problem solved by $\cal A$.
His improvement comes from
an observation that performing a separate binary search for each step of the
algorithm $\cal A$ will often ``waste'' calls to $\cal C$ to resolve a
relatively small number of comparisons.
Rather than resolve all the comparisons of a single step of $\cal A$, he
instead assumes that $\cal A$ is implemented as 
the AKS sorting network~\cite{Ajtai:1983:SCL:61981.61982}
or an optimal EREW parallel sorting algorithm~\cite{Cole:1988,Goodrich:1996},
which allows for comparisons on multiple steps of $\cal A$ to
be considered at the same time (so long as their preceding comparisons have
been resolved).
This improvement results in a running time for the optimization
problem that is 
$O(P(n)T(n) + C(n) (T(n)+\log P(n)))$.

From an algorithm engineering perspective, the ``classical''
parametric search technique 
(utilizing a parallel algorithm)
is admittedly 
difficult to implement, although some implementations 
do exist~\cite{DBLP:conf/compgeom/SchwerdtSS97,Toledo91extremalpolygon,
DBLP:journals/comgeo/OostrumV04}.
Cole's improvement is even more complex, however, and
we are not familiar with any implementations of his 
parametric search optimization.

Even without Cole's improvement, a challenge for implementing
the parametric search technique is the simulation
of a parallel algorithm on a sequential machine.
This difficulty has motivated some researchers to abandon 
the use of parametric searching
entirely and instead use other paradigms,
such as expander graphs~\cite{Katz1993115}, 
geometric random sampling~\cite{Mat91},
and $\epsilon$-cuttings~\cite{Bronnimann199823}
(see also~\cite{Agarwal:1998}).

Interestingly,
 van~Oostrum and Veltkamp~\cite{DBLP:journals/comgeo/OostrumV04}
show that, for sorting-based parametric search applications, one can 
use the well-known \texttt{quicksort} algorithm\ to 
drive comparisons instead of a parallel sorting algorithm.
Unfortunately, as van~Oostrum and Veltkamp note in their paper, Cole's 
improvement cannot be applied in this case.
The main difficulty is that, when viewed as a kind of parallel algorithm,
comparisons to be done at one level of \texttt{quicksort} become known
only after all the comparisons on the level above have been resolved.
Thus, comparisons
cannot be pipelined in the way required by Cole's optimization when using
this approach.
The result, of course, is that this sets up an unfortunate tension between
theory and practice, forcing algorithm designers to choose between
a practical, but asymptotically inferior, implementation or an
impractical
algorithm whose running time is asymptotically better by a logarithmic
factor.

\subsection{Our Results}
We show
that it is, in fact,
possible to implement Cole's parametric search technique
in a manner that is efficient and
practical (i.e., fast and easy to implement).
The main idea is to use a variant
of quicksort, known as 
\texttt{boxsort}~\cite{DBLP:journals/siamcomp/Reischuk85},
to drive comparisons (instead of sorting networks,
like the complicated AKS network or an EREW parallel
sorting algorithm).
We apply a potential 
function to comparisons in the \texttt{boxsort} 
algorithm, which,
together with a weighted-median-finding algorithm,
allows us to schedule these comparisons in a pipelined fashion
and achieve, with high probability,
the same asymptotic running time as Cole's method,
while also being practical. 
Moreover, we provide experimental results that
give empirical evidence supporting these claims
for the ``median-of-lines'' problem~\cite{DBLP:journals/jacm/Megiddo83} and
the geometric optimization problems of matching planar drawings~\cite{Alt2003262}
and labeling planar maps with rectangles~\cite{knntw-lprvs-02}.

\section{Parametric Search Explained} \label{sec_psearch}

In this section,
we provide a more in-depth description of the parametric search technique.
Recall that $B$ is a problem that we want to solve.
Furthermore, we restrict ourselves to the case
where the generic algorithm $\cal A$ is a sorting algorithm.
We require of $B$ the following.

\begin{enumerate}
  \item There is a \emph{decision algorithm}, $\cal C$,
    which, for any value $\lambda$,
    resolves a comparison $\lambda < \lambda^*$
    in time $C(n)$ without actually knowing $\lambda^*$
    (note that $C(n)$ is a function of the size of input to $B$).
    Typically, $C(n)$ is at least $\Omega(n)$, as opposed
    to $O(1)$ comparison time which is usual for classical sorting algorithms.
  \item There is an efficient way of generating values $x_i$
    (with each $x_i$ being either a real value or a real-valued function of $\lambda$)
    from an input to problem $B$. Ideally, it produces $O(n)$ such values.
  \item For each $x_i<x_j$ comparison,
    the answer is determined by the sign of a low-degree polynomial in $\lambda$
    at $\lambda = \lambda^*$ (polynomials for different comparisons may differ).
  \item \emph{Critical values} (values $\lambda$ that, based on combinatorial
    properties of $B$, have the potential of being equal to $\lambda^*$)
    form a subset of the set of roots of the polynomials determining answers
    to every possible comparison $x_i<x_j$.
\end{enumerate}

Then, as a by-product of sorting values $x_i$, we get (directly or indirectly)
the answers to all comparisons $\lambda < \lambda^*$, where $\lambda$'s are roots
of all comparisons $x_i < x_j$. Therefore, we are able to find $\lambda^*$.

We can solve $B$ in the following way: generate
$x_i$'s, sort them
using algorithm $\mathcal{A}$ and recover $\lambda^*$ from the answer.
If $\mathcal{A}$ sorts $n$ items in $T(n)$ comparisons and each
comparison is resolved in time $O\big(C(n)\big)$ (it requires determining
whether $\lambda < \lambda^*$ for a constant number of roots $\lambda$),
solving $B$ this way takes time $T(n)C(n)$.

It is important to note that if there are $k$ comparisons $x_i < x_j$,
we can avoid calling $\mathcal{C}$ on every single root of their polynomials,
and still resolve them all.
This is because resolving $\lambda < \lambda^*$ automatically resolves
comparisons for values $\lambda' \leq \lambda$ (if the result was
\verb\YES\) or $\lambda'' > \lambda^*$ (if the result was \verb\NO\).
Therefore, we can solve $k$ comparisons in only $O(\log k)$ calls to $\mathcal{C}$,
if in every iteration we use a standard median-finding algorithm (e.g., see \cite{clrs})
to find the median root $\lambda$, and then resolve it by a call to $\mathcal{C}$
(each iteration halves the number of unresolved comparisons).

The above observation lies at the heart of the original
parametric search, as introduced by Megiddo~\cite{DBLP:journals/jacm/Megiddo83}.
Note that we can group the comparisons in such a way only if they
are \emph{independent} of each other.
To assure this, one chooses $\mathcal{A}$ to be
a \emph{parallel} sorting algorithm, running in time $T(n)$ on $P(n)$
processors. At every step of $\mathcal{A}$, there are $O(P(n))$
independent comparisons, and they can be resolved in time
$O\big(P(n) + \log(P(n))\cdot C(n)\big)$ according to the previous observation.
Resolving comparisons at all $T(n)$ steps of $\mathcal{A}$ takes time
$O\big(T(n)\cdot P(n) + T(n)\cdot\log(P(n))\cdot C(n)\big)$.
Simulating $\mathcal{A}$ on a sequential machine takes time
$O\big(T(n) P(n)\big)$.
Therefore, parametric search, as originally introduced, helps solve $B$ in time
$O\big(T(n)\cdot P(n) + T(n)\cdot\log(P(n))\cdot C(n)\big)$.

\subsection{Cole's Improvement}

Cole~\cite{DBLP:journals/jacm/Cole87} was able to improve on
Megiddo's result by using a sorting network or an EREW parallel
sorting algorithm as
$\mathcal{A}$, and changing
the order of comparison resolution by assigning weights to
comparisons and resolving the \emph{median weighted comparison}
at each step.

In the case of a sorting network,
a straightforward notion of \emph{active} comparisons
and \emph{active} wires was introduced. Initially, all input wires
(and no others) are \emph{active}. A comparison is said to be
\emph{active} if it is not resolved and both its input wires are \emph{active}.
When active comparison gets resolved, its output wires now become
\emph{active}, possibly activating subsequent comparisons. Informally,
\emph{active} comparisons have not been resolved yet, but both of their
inputs are already determined.

\emph{Weight} is assigned to every comparison, being
equal to $4^{-j}$ for a comparison at depth $j$. The \emph{active weight}
is defined as the weight of all \emph{active} comparisons. The weighted
median comparison can be found in $O(n)$ time~\cite{DBLP:journals/ipl/Reiser78},
and resolving it automatically resolves a weighted half of the comparisons.

It is shown that for a sorting network of width $P(n)$ and depth $T(n)$,
or an EREW sorting algorithm with $P(n)$ processors and time $T(n)$,
the method of resolving weighted median comparison requires only
$O(T(n) + \log( P(n)))$ direct calls to $\mathcal{C}$.
Including simulation overhead, we solve $B$ in time
$O(P(n)\cdot T(n) + \big(T(n) + \log(P(n))\big)\cdot C(n))$.

This is completely impractical, however, as the bounds for the AKS network
have huge constant factors.
In a subsequent work~\cite{Cole:1988}, Cole shows that one can substitute
an EREW parallel sorting algorithm for the AKS network, which makes using his
optimization more implementable, but arguably still not practical, since
the existing optimal EREW parallel
sorting algorithms~\cite{Cole:1988,Goodrich:1996}
are still fairly intricate.

\cprotect\subsection{Applying \verb/quicksort/ to Parametric Search}

Van~Oostrum and Veltkamp~\cite{DBLP:journals/comgeo/OostrumV04} have shown
that the \verb\quicksort\ algorithm~\cite{Hoare:1961}
can be used as $\mathcal{A}$.
Recall that in the randomized version of this algorithm we sort a
set of elements by picking one of them (called the \emph{pivot}) at random, and
recursively sorting elements smaller than the pivot and greater than the pivot.
A key observation here is that all
the comparisons with the pivot(s) at a given level of recursion are
independent of each other.
It leads to a practical algorithm, running
in $O(n\log n + \log^2n\cdot C(n))$
expected-time, for solving $B$ (it becomes
$O(n\log n + \log n\cdot C(n))$ under additional
assumption about distribution of the roots of polynomials). Comparisons
are resolved
by resolving the median comparison among unresolved
comparisons at the current level. As \verb\quicksort\ is
expected to have $O(\log n)$ levels of recursion, and $O(n)$
comparisons at each level can be resolved in time
$O(n + \log n\cdot C(n))$, time bound follows.

Cole's improvement cannot be applied in this case, because all
comparisons at one level have to be resolved before we even know
what comparisons have to be done at the next level (that is, we don't know
the splits around pivots until the very last comparison is resolved).


\section{Our Practical Version of Cole's Technique}
\label{sec_our_psearch}
In this section, we describe our algorithm engineering
framework for making Cole's parametric search
technique practical.
Our approach results in a randomized
parametric search algorithm with a
running time of $O(n\log n + \log n\cdot C(n))$, with high probability,
which makes no assumptions about the input.
Our framework involves resolving median-weight comparison,
according to a potential function based on
Cole-style weights assigned to comparisons
of a fairly obscure sorting algorithm, which we review next.

\cprotect\subsection{The \verb\boxsort\ Algorithm}
\label{boxsort_label}

We use the \verb\boxsort\ algorithm due to
Reischuk~\cite{DBLP:journals/siamcomp/Reischuk85} (see also~\cite{Motwani:1995:RA:211390})
as $\mathcal{A}$.
This algorithm
is based on an extension of the main idea behind
randomized \verb\quicksort\, namely
splitting elements around pivots and recursing
into subproblems. While \verb\quicksort\ randomly selects a single
pivot and recurses into two subproblems, \verb\boxsort\ randomly
selects $\sqrt{n}$ pivots and recurses into $\sqrt{n}+1$ subproblems
in a single stage. We think of it as a parallel algorithm, in the sense
that the recursive calls on the same level are independent of each other.
The pseudocode is shown in Algorithm~\ref{alg_boxsort}.

\begin{algorithm}
  \begin{algorithmic}[1]
    \ENSURE \emph{// $N$ -- original number of items}
    \REQUIRE \texttt{boxsort($A[i\ldots j]$)}
    \STATE $n\leftarrow (j-i+1)$
    \IF[\emph{$\qquad$// base case}]{$n < \log N$}
      \STATE sort $A[i\ldots j]$
    \ELSE
      \STATE randomly \emph{mark} $\sqrt{n}$ items
      \STATE sort the \emph{marked} items
      \STATE use the \emph{marked} items to split
        $A[i\ldots j]$  into subproblems $A_1, A_2, \ldots, A_{\sqrt{n}+1}$
      \FORALL{$i\leftarrow1\ldots \sqrt{n}+1$}
        \STATE \texttt{boxsort($A_i$)}
      \ENDFOR
    \ENDIF
  \end{algorithmic}
  \caption{\texttt{boxsort}}
  \label{alg_boxsort}
\end{algorithm}

Few details need further explanation. Sorting
in lines~3 and~6 is done in a brute-force manner, by comparing
all pairs of items, in time $O(n^2)$ in line~3,
and $O(n)$ in line~6 (note that since all these comparisons are
independent, they can all be realized in a single parallel step).

Once the \emph{marked} items are sorted in line~6,
splitting in line~7 is simply done by $n - \sqrt{n}$ independent binary searches
through the \emph{marked} items (to determine, for each unmarked element,
the subproblem where it lands).
It takes $O(n\log\sqrt{n})$ time (when realized in a sequential way).
Equivalently, we think of the sorted set of \emph{marked}
items as forming a perfectly balanced binary search tree.
Locating a destination subproblem for an item is then done by
\emph{routing} the item through this tree.
The tree has $\log\sqrt{n}$ levels, and all routing comparisons
are independent between different unmarked items.
Therefore, \emph{routing} can be realized in $\log\sqrt{n}$
parallel steps.

\subsection{Weighting Scheme}

Motivated by Cole's approach,
we assign weight to every \emph{active}
comparison, and resolve the
weighted median comparison in a single step. For simplicity, we identify each
comparison $x_i < x_j$ with a single comparison against the optimum
value, i.e., $\lambda_{ij} < \lambda^*$ for real $\lambda_{ij}$
(in essence, we assume that comparison polynomials have degree 1).
It is straightforward to extend the scheme for the case
of higher degrees of comparison polynomials.

It makes sense here to think of \verb\boxsort\ in a network-like fashion,
in order to understand how the weights are assigned to comparisons.
Here, nodes represent comparisons, and directed edges represent dependence
on previous comparisons. Furthermore, we imagine the network with
edges directed downward, and refer to edge sources as \emph{parents},
and destinations as \emph{children}.
Comparison becomes \emph{active} as soon as
all its dependencies become resolved (and stops when
it gets resolved).

Our ``network'' also contains nodes
for \emph{virtual comparisons}. These are not
real comparisons, and don't appear during actual execution of the
algorithm. Their sole purpose is to make it easy to assign weights to
\emph{real} comparisons once they become \emph{active} (we will later see
that, in fact, they are not necessary even for that; but they
make it easy to understand how the weights are computed).
When a \emph{virtual} comparison becomes \emph{active}, it is automatically
resolved (reflecting the fact that there is
no \emph{real} work assigned to a virtual comparison).

Contrary to Cole's weighting scheme for sorting networks, our scheme
does not rely only on comparison's depth when assigning weights.
In fact, different comparisons at the same level of
the network may have different weights. Weights are assigned
to comparisons (virtual or not) according to the following
\emph{weight rule}:

\begin{center}
\fbox{
  \parbox{2.8in}{
    When comparison $C$ of weight $w$ gets resolved
    and causes $m$ comparisons $C_1,\ldots,C_m$ to become
    \emph{active}, each of these comparisons
    gets weight $w/2m$.
  }
}
\end{center}

Informally, resolved comparison distributes half
of its weight among its newly activated children.
Each comparison gets its
weight only once, from its last resolved parent
(the scheme guarantees that all parents of a
comparison have equal weight).

\subsection{The Algorithm}
\label{algorithm_label}

Simulating a single recursive call of \verb\boxsort\
(including the \emph{virtual} parts) consists of the following steps.

\begin{enumerate}
  \item Randomly \emph{mark} $\sqrt{n}$ items.
  \item Create $\sqrt{n}\cdot(\sqrt{n}-1)/2 = O(n)$ comparisons for
    sorting \emph{marked} items.
  \textcolor{blue}{\item Construct a complete binary tree of virtual comparisons
    (comparisons from Step~2 are leaves).}
  \textcolor{red}{\item Create \emph{routing} trees from
    section~\ref{boxsort_label} for routing unmarked elements;
    make the root of each such tree depend on the root of the tree from Step~3.}
  \item Route items through the tree of \emph{marked} items;
  \textcolor{blue}{\item Construct a binary tree of virtual comparisons
    (leaves are last comparisons from \emph{routing} trees).}
  \item Split items into boxes
  \textcolor{red}{\item Assign weights for comparisons
    in the next level of recursion (after the items are split
    into boxes) by making them children of the root from Step~6.}
  \item Recurse into subproblems (simultaneously).
\end{enumerate}

Colors represent \emph{virtual} parts of the algorithm and correspond
to the pictorial explanation of the algorithm from the appendix
(Figure~\ref{fig_1}).
\textcolor{blue}{Blue} steps (3, 6) deal with trees of virtual comparisons,
while \textcolor{red}{red} steps (4, 8) represent relationships that
make \emph{real} comparisons depend on \emph{virtual} ones.
The idea behind \textcolor{blue}{blue} steps is to ensure synchronization
(that is, guarantee that all \emph{real} comparisons on the levels above have been
resolved), and \textcolor{red}{red} steps are there to ensure proper
assignment of weights.
Figure~\ref{fig_1} also shows
heights of the trees used and weights assigned to comparisons on levels
of the ``network''. For simplicity, it presents heights/weights as if there were
exactly $n$ (instead of $\sqrt{n}\cdot(\sqrt{n}-1)/2$)
comparisons between \emph{marked} items, and exactly $n$
(instead of $n - \sqrt{n}$)
unmarked items to be routed. The following discussion is also based on this assumption.

Steps 1 and 7 do not involve any comparisons, and they
do not affect weights.
Comparisons from Step~2 start with weight $w$.
The tree from \textcolor{blue}{Step~3} has height $\log n$, so its root,
according to the \emph{weight rule} gets weight $w / (2^{\log{n}}) = w/n$.
Dependencies introduced in \textcolor{red}{Step~4} between that root and
roots of the \emph{routing trees} cause their weight to be $w/2n^2$
(weight $w/n$ divided among $n$ comparisons). \emph{Routing trees} have height
$\log{\sqrt{n}}$, so the comparisons at their bottom have weight $w/2n^{2.5}$
($w/2n^2$ divided by $2^{\log{\sqrt{n}}}$, because, as the routing progresses,
the \emph{routing trees} get whittled down to paths, and resolving a routing comparison
\emph{activates} at most one new routing comparison.
\textcolor{blue}{Step~6} is essentially the same as
\textcolor{blue}{Step~3}, so the root of the second $\emph{virtual}$
tree gets weight $w/2n^{3.5}$. All initial comparisons in the subsequent
recursive calls (sorting of new \emph{marked} items and/or sorting
in the base case) depend on this root (\textcolor{red}{Step~8}),
and they are given weight $w/4n^{4.5}$
(much like in \textcolor{red}{Step~4}). The height of the dependence
network is $O(\log{n})$, and at any given moment the number of
currently \emph{active} comparisons does not exceed $n$.

From now on, comparisons are independent across different subproblems.
For subsequent subproblems, $n$ from the above discussion
gets substituted by $\hat{n}$, the size of the subproblem.
Since subproblem sizes may differ, comparisons
on the same level of the  network (general level, for the entire algorithm)
are no longer guaranteed to have same weights
(weights of comparisons belonging to the same subproblem are however equal).

The above discussion shows that, as advertised, we don't really need \emph{virtual}
comparisons in order to assign weights to \emph{real} comparisons, as these
depend only on $n$, the size of the subproblem.
Therefore, the actual algorithm only consists of steps 1, 2, 5, 7, and 9
and is the following.

\begin{enumerate}
  \item Randomly \emph{mark} $\sqrt{n}$ items
  \item Sort \emph{marked} items by comparing
    every pair in $O(n)$
    comparisons, each of weight $w$.
  \item When the last comparison finishes,
    \emph{activate} comparisons for routing through
    the tree of \emph{marked} items, each of weight $w/2n^2$.
  \item Route items through the trees, following the
    \emph{weight rule} when a comparison gets resolved.
  \item When the destination for the last item
    is determined, split items into boxes
    (no additional comparisons resolved here).
  \item Assign weight $w/4n^{4.5}$ to initial
    comparisons in new subproblems.
  \item Recurse into subproblems (simultaneously).
\end{enumerate}

\subsection{Analysis}
\label{analysis_label}

Assume that initially all comparisons at the highest level
were given weight 1. In this analysis, we also include
\emph{virtual} comparisons.
If the current \emph{active} weight (sum of weights
of all \emph{active} comparisons) is equal to $W$,
resolving the weighted-median comparison reduces
\emph{active} weight by at least $W/4$ (it resolves
comparisons of total weight $\geq W/2$, and
each resolved comparison passes only at most half of
its weight to its children). Thus,
the following lemma is proved identically
as Lemma~1 of~\cite{DBLP:journals/jacm/Cole87}
(a turn consists of resolving median weighted comparison
and assigning weights to newly activated comparisons).

\begin{lemma}
At the start of the (k+1)-st turn, active weight is bounded
from above by $(3/4)^kn$, for $k\geq0$.
\end{lemma}

We also have the following.

\begin{lemma}
\label{weight_lemma}
Each comparison at depth $j$ has weight $\geq 4^{-j}$.
\end{lemma}
\begin{proof}
We prove this by induction on the depth of the \verb\boxsort\
recursion.
Assume that the current recursive call operates on
a subproblem of size $n$, comparisons at the beginning
of the current recursive call have depth $k$ and weight $w$.
By inductive assumption, $w \geq 4^{-k}$.

Consider comparisons in the current recursive
call. Comparisons at depth $i$ in the first tree
of \emph{virtual} comparisons (global depth $k+i$)
have weight
$w/2^i \geq 4^{-k}\cdot2^{-i} \geq 4^{-(k+i)}$.
The last of them has (local) depth $\log n$ and weight $w/n$.
It then spreads half of its weight to $n$ comparisons at
depth $\log n +1$ (global depth $k+\log n+1$), setting their weight to
$w/2n^2 \geq w/4n^2 = w/4^{\log n + 1} \geq 4^{-(k + \log n + 1)}$.
The same reasoning follows for the case of the second
\emph{virtual} tree and recursive split, while routing
through the tree of sorted \emph{marked} items
always decreases weight by a factor of 2 for the
next level instead of 4 (making the result even stronger).

To finish the proof, note that the base
case is realized in the very first call
to the algorithm,
since comparison at depth 0 has weight $1 = 1/4^0$.
\end{proof}

Lemma~\ref{weight_lemma} allows us to prove the following
lemma exactly as Lemma 2 of \cite{DBLP:journals/jacm/Cole87}.

\begin{lemma}
For $k\geq5(j+1/2\cdot\log n)$, during the (k+1)-st turn
there are no active comparisons at depth $j$.
\end{lemma}

This leads to the following corollary.

\begin{lemma}
\label{depth_lemma}
If the \emph{network} for \verb\boxsort\ has height $f(n)$,
$O(f(n) + \log n)$ rounds of resolving the median-weight
comparison suffice to resolve every comparison in the network.
\end{lemma}

We also have
the following fact about \verb\boxsort\.

\begin{lemma} (Theorem 12.2~of~\cite{Motwani:1995:RA:211390})
\label{book_lemma}
There is a constant $b > 0$ such that
\verb\boxsort\ terminates in $O(\log n)$ parallel steps
with probability at least $1 - \exp(-\log^bn)$.
\end{lemma}

Originally, \verb\boxsort\ requires $O(\log n)$ parallel steps
to execute a single recursive call for a problem of size $n$.
We noted that the dependence network for a single recursive call
in our simulation has height $O(\log n)$ for a problem of size $n$
as well. This means that Lemma~\ref{book_lemma} applies here
and proves that,
with high probability, the dependence network for the
entire simulation has height $O(\log n)$.

Combining that with Lemma~\ref{depth_lemma} and the
observation that any level in the dependence network
contains $O(n)$ comparisons, we get the following.

\begin{theorem}
\label{main_thm}
With high probability, the presented algorithm requires
$O(\log n)$ calls to $\cal C$, yielding an
$O(n\log n +\log n\cdot C(n))$ time parametric search
solution to problem $B$.
\end{theorem}

\section{Conclusion} \label{sec_conclusion}
We have introduced a practical version of Cole's optimization
of the parametric search technique.
Our method results in a randomized algorithm whose running time matches
that of using Cole's technique, with high probability, while
being easily implementable.
We have implemented it and, based on experimentation performed
on some geometric problems (details in the appendix),
showed that our approach is competitive with
the previous practical parametric search technique of
van~Oostrum and Veltkamp~\cite{DBLP:journals/comgeo/OostrumV04}, while
having superior asymptotic performance guarantees.

\bibliographystyle{abbrv}
\bibliography{refs}

\newpage
\appendix

\section{Appendix}

\subsection{Pictorial Illustration of the Algorithm}

\begin{figure*}[ht] 
  \begin{center}
    \begin{tikzpicture}
      [input/.style={inner sep=1pt,circle,fill=black},
       input_virt/.style={inner sep=1pt,circle,fill=blue},
      input2/.style={inner sep=0.75pt,circle,fill=black},
      input2_virt/.style={inner sep=0.75pt,circle,fill=blue},
      spread/.style={red,very thin},
      virt/.style={blue,dashed},
      triangle/.style={black,fill=black!45},
      line1/.style={very thin,loosely dashed},
      line2/.style={very thin,dotted},
      empty/.style={inner sep=0pt}]

      \node at (0,0.5) {\scriptsize weights};


      \node[text width=2.7cm,text centered] (desc) at (11,0) {\scriptsize $n$ comparisons for $\sqrt{n}$ \emph{marked} elements};

      \path (1.5,0.0) node (a) [input] {}
            (2.0,0.0) node (b) [input] {}
            (2.5,0.0) node (c) [input] {}
            (3.0,0.0) node (d) [input] {}
            (3.5,0.0) node (e) [input] {}
            (4.0,0.0) node (f) [input] {}
            (5.25,0) node  {$\ldots$}
            (6.5,0.0) node (g) [input] {}
            (7.0,0.0) node (h) [input] {}
            (7.5,0.0) node (i) [input] {}
            (8.0,0.0) node (j) [input] {}
            (8.5,0.0) node (k) [input] {}
            (9.0,0.0) node (l) [input] {};

      \node (w) at (0,0) {\scriptsize $w$};
      \draw[line2] (w) -- (a);
      \draw[line2] (l) -- (desc);

      \node[input_virt] (ab) at (1.75,-0.25) {};
      \node[input_virt] (cd) at (2.75,-0.25) {};
      \node[input_virt] (ef) at (3.75,-0.25) {};
      \node[input_virt] (gh) at (6.75,-0.25) {};
      \node[input_virt] (ij) at (7.75,-0.25) {};
      \node[input_virt] (kl) at (8.75,-0.25) {};

      \draw[line1] (0.8,-0.25) -- (ab);
      \node (w) at (0,-0.25) {\scriptsize $w/2$};
      \draw[line2] (w) -- (0.8,-0.25);

      \coordinate (start) at (0.8,-0.25);

      \draw[virt] (a) -- (ab);
      \draw[virt] (b) -- (ab);
      \draw[virt] (c) -- (cd);
      \draw[virt] (d) -- (cd);
      \draw[virt] (e) -- (ef);
      \draw[virt] (f) -- (ef);
      \draw[virt] (g) -- (gh);
      \draw[virt] (h) -- (gh);
      \draw[virt] (i) -- (ij);
      \draw[virt] (j) -- (ij);
      \draw[virt] (k) -- (kl);
      \draw[virt] (l) -- (kl);

      \node (t) [input_virt] at (5.25,-4.0) {};

      \draw[line1] (0.8,-4.0) -- (t);
      \coordinate (end) at (0.8,-4.0);

      { [>=stealth']
        \draw[<->] (start) -- (end);
        \node[fill=white] at ($(start)!0.5!(end)$) {\scriptsize $\log n$};
        \node (w) at ($(end) - (0.8,0)$) {\scriptsize $w/n$};
        \node[blue] (desc) at ($(start)!0.5!(end) + (10.2,0)$) {\scriptsize \emph{virtual} comparisons};
        \node[blue] (desc) at ($(w) + (11,0)$) {\scriptsize last \emph{virtual} comparison};
        \draw[line2] (t) -- (desc);
        \draw[line2] (w) -- (end);
      }

      \draw[virt] (ab) -- (kl) -- (t) -- (ab);

      { [yshift=-4.75cm]; 
        \path (1.5,0.0) node (a) [input] {}
              (2.0,0.0) node (b) [input] {}
              (2.5,0.0) node (c) [input] {}
              (3.0,0.0) node (d) [input] {}
              (3.5,0.0) node (e) [input] {}
              (4.0,0.0) node (f) [input] {}
              (5.25,0) node  {$\ldots$}
              (6.5,0.0) node (g) [input] {}
              (7.0,0.0) node (h) [input] {}
              (7.5,0.0) node (i) [input] {}
              (8.0,0.0) node (j) [input] {}
              (8.5,0.0) node (k) [input] {}
              (9.0,0.0) node (l) [input] {};

        \draw[line1] (0.8,0.0) -- (a);
        \node (w) at (0.0,0.0) {\scriptsize $w/2n^2$};
        \draw[line2] (w) -- (0.8,0.0);

        \coordinate (start) at (0.8,0.0);

        \draw[spread] (t) -- (a);
        \draw[spread] (t) -- (b);
        \draw[spread] (t) -- (c);
        \draw[spread] (t) -- (d);
        \draw[spread] (t) -- (e);
        \draw[spread] (t) -- (f);
        \draw[spread] (t) -- (g);
        \draw[spread] (t) -- (h);
        \draw[spread] (t) -- (i);
        \draw[spread] (t) -- (j);
        \draw[spread] (t) -- (k);
        \draw[spread] (t) -- (l);

        \node [below=-0.1cm of a] {\tiny $\Downarrow$};
        \node [below=-0.1cm of b] {\tiny $\Downarrow$};
        \node [below=-0.1cm of c] {\tiny $\Downarrow$};
        \node [below=-0.1cm of d] {\tiny $\Downarrow$};
        \node [below=-0.1cm of e] {\tiny $\Downarrow$};
        \node [below=-0.1cm of f] {\tiny $\Downarrow$};
        \node [below=-0.1cm of g] {\tiny $\Downarrow$};
        \node [below=-0.1cm of h] {\tiny $\Downarrow$};
        \node [below=-0.1cm of i] {\tiny $\Downarrow$};
        \node [below=-0.1cm of j] {\tiny $\Downarrow$};
        \node [below=-0.1cm of k] {\tiny $\Downarrow$};
        \node [below=-0.1cm of l] {\tiny $\Downarrow$};

        { [yshift=-0.3cm]
          \coordinate (a) at (1.5,0.0);
          \coordinate (b) at (2.0,0.0);
          \coordinate (c) at (2.5,0.0);
          \coordinate (d) at (3.0,0.0);
          \coordinate (e) at (3.5,0.0);
          \coordinate (f) at (4.0,0.0);
          \coordinate (g) at (6.5,0.0);
          \coordinate (h) at (7.0,0.0);
          \coordinate (i) at (7.5,0.0);
          \coordinate (j) at (8.0,0.0);
          \coordinate (k) at (8.5,0.0);
          \coordinate (l) at (9.0,0.0);
          \coordinate (la) at (1.3,-0.5);
          \coordinate (ra) at (1.7,-0.5);
          \coordinate (lb) at (1.8,-0.5);
          \coordinate (rb) at (2.2,-0.5);
          \coordinate (lc) at (2.3,-0.5);
          \coordinate (rc) at (2.7,-0.5);
          \coordinate (ld) at (2.8,-0.5);
          \coordinate (rd) at (3.2,-0.5);
          \coordinate (le) at (3.3,-0.5);
          \coordinate (re) at (3.7,-0.5);
          \coordinate (lf) at (3.8,-0.5);
          \coordinate (rf) at (4.2,-0.5);
          \coordinate (lg) at (6.3,-0.5);
          \coordinate (rg) at (6.7,-0.5);
          \coordinate (lh) at (6.8,-0.5);
          \coordinate (rh) at (7.2,-0.5);
          \coordinate (li) at (7.3,-0.5);
          \coordinate (ri) at (7.7,-0.5);
          \coordinate (lj) at (7.8,-0.5);
          \coordinate (rj) at (8.2,-0.5);
          \coordinate (lk) at (8.3,-0.5);
          \coordinate (rk) at (8.7,-0.5);
          \coordinate (ll) at (8.8,-0.5);
          \coordinate (rl) at (9.2,-0.5);

          \draw[triangle] (a) -- (la) -- (ra) -- cycle;
          \draw[triangle] (b) -- (lb) -- (rb) -- (b);
          \draw[triangle] (c) -- (lc) -- (rc) -- (c);
          \draw[triangle] (d) -- (ld) -- (rd) -- (d);
          \draw[triangle] (e) -- (le) -- (re) -- (e);
          \draw[triangle] (f) -- (lf) -- (rf) -- (f);
          \draw[triangle] (g) -- (lg) -- (rg) -- (g);
          \draw[triangle] (h) -- (lh) -- (rh) -- (h);
          \draw[triangle] (i) -- (li) -- (ri) -- (i);
          \draw[triangle] (j) -- (lj) -- (rj) -- (j);
          \draw[triangle] (k) -- (lk) -- (rk) -- (k);
          \draw[triangle] (l) -- (ll) -- (rl) -- (l);

          \draw[line1] (0.8,-0.5) -- (la);
          \draw[line2] (0.0,-0.5) -- (0.8,-0.5);

          \coordinate (end) at (0.8,-0.5);
          \node[fill=white] at (0.0,-0.5) {\scriptsize $w/2n^{2.5}$};
          { [>=stealth']
            \draw[<->] (start) -- (end);
            \node[fill=white] at ($(start)!0.5!(end)$) {\scriptsize $\log \sqrt{n}$};
            \node[text width=2.6cm,text centered] at ($(start)!0.5!(end) + (10.2,0)$)
              {\scriptsize routing all elements through \emph{marked} elements};
          }

          { [yshift=-0.8cm] 
            \path (1.5,0.0) node (a) [input_virt] {}
                  (2.0,0.0) node (b) [input_virt] {}
                  (2.5,0.0) node (c) [input_virt] {}
                  (3.0,0.0) node (d) [input_virt] {}
                  (3.5,0.0) node (e) [input_virt] {}
                  (4.0,0.0) node (f) [input_virt] {}
                  (6.5,0.0) node (g) [input_virt] {}
                  (7.0,0.0) node (h) [input_virt] {}
                  (7.5,0.0) node (i) [input_virt] {}
                  (8.0,0.0) node (j) [input_virt] {}
                  (8.5,0.0) node (k) [input_virt] {}
                  (9.0,0.0) node (l) [input_virt] {};

            \draw[line1] (0.8,0.0) -- (a);
            \node (w) at (0.0,0.0) {\scriptsize $w/2n^{2.5}$};
            \draw[line2] (w) -- (0.8,0.0);

            \coordinate (start) at (0.8,0.0);

            \node [above=-0.1cm of a] {\tiny $\Downarrow$};
            \node [above=-0.1cm of b] {\tiny $\Downarrow$};
            \node [above=-0.1cm of c] {\tiny $\Downarrow$};
            \node [above=-0.1cm of d] {\tiny $\Downarrow$};
            \node [above=-0.1cm of e] {\tiny $\Downarrow$};
            \node [above=-0.1cm of f] {\tiny $\Downarrow$};
            \node [above=-0.1cm of g] {\tiny $\Downarrow$};
            \node [above=-0.1cm of h] {\tiny $\Downarrow$};
            \node [above=-0.1cm of i] {\tiny $\Downarrow$};
            \node [above=-0.1cm of j] {\tiny $\Downarrow$};
            \node [above=-0.1cm of k] {\tiny $\Downarrow$};
            \node [above=-0.1cm of l] {\tiny $\Downarrow$};

            \node[input_virt] (ab) at (1.75,-0.25) {};
            \node[input_virt] (cd) at (2.75,-0.25) {};
            \node[input_virt] (ef) at (3.75,-0.25) {};
            \node[input_virt] (gh) at (6.75,-0.25) {};
            \node[input_virt] (ij) at (7.75,-0.25) {};
            \node[input_virt] (kl) at (8.75,-0.25) {};

            \draw[virt] (a) -- (ab);
            \draw[virt] (b) -- (ab);
            \draw[virt] (c) -- (cd);
            \draw[virt] (d) -- (cd);
            \draw[virt] (e) -- (ef);
            \draw[virt] (f) -- (ef);
            \draw[virt] (g) -- (gh);
            \draw[virt] (h) -- (gh);
            \draw[virt] (i) -- (ij);
            \draw[virt] (j) -- (ij);
            \draw[virt] (k) -- (kl);
            \draw[virt] (l) -- (kl);

            \node (t) [input_virt] at (5.25,-4.0) {};

            \draw[line1] (0.8,-4.0) -- (t);
            \node (w) at (0.0,-4.0) {\scriptsize $w/2n^{3.5}$};
            \draw[line2] (w) -- (0.8,-4.0);
            \coordinate (end) at (0.8,-4.0);
            { [>=stealth']
              \draw[<->] (start) -- (end);
              \node[fill=white] at ($(start)!0.5!(end)$) {\scriptsize $\log n$};
              \node[blue] (desc) at ($(start)!0.5!(end) + (10.2,0)$) {\scriptsize \emph{virtual} comparisons};
              \node[blue] (desc) at ($(w) + (11,0)$) {\scriptsize last \emph{virtual} comparison};
              \draw[line2] (t) -- (desc);
            }

            \draw[virt] (a) -- (t) -- (l) -- (a);
            \node [fill=white] at (5.25,0) {$\ldots$};

            { [yshift=-5.0cm] 
              \draw[green,thick,dashed,fill=green!15] (3.25,0.0) circle [x radius=1.0cm, y radius=0.2cm];

              \path (1.0,0.0) node (ls) [input] {}
                    (2.5,0.0) node [input] {}
                    (3.25,0.0) node (rs) {$\hat{n}$}
                    (4.0,0.0) node [input] {}
                    (5.25,0.0) node  {$\ldots$}
                    (6.5,0.0) node [input] {}
                    (8.0,0.0) node [input] {}
                    (9.5,0.0) node [input] {};

              \node (w) at (0.0,0.0) {\scriptsize $w'=w/4n^{4.5}$};
              \node[red] (desc) at (11,0) {\scriptsize split into \emph{boxes}};
              \draw[line2] (w) -- (1.0,0.0);

              \draw[spread] (t) -- (1.2,0.0);
              \draw[spread] (t) -- (1.6,0.0);
              \draw[spread] (t) -- (2.0,0.0);
              \draw[spread] (t) -- (3.2,0.2);
              \draw[spread] (t) -- (3.6,0.2);
              \draw[spread] (t) -- (4.0,0.2);
              \draw[spread] (t) -- (4.0,0.0);
              \draw[spread] (t) -- (6.5,0.0);
              \draw[spread] (t) -- (6.9,0.0);
              \draw[spread] (t) -- (7.3,0.0);
              \draw[spread] (t) -- (7.7,0.0);
              \draw[spread] (t) -- (8.1,0.0);
              \draw[spread] (t) -- (8.5,0.0);
              \draw[spread] (t) -- (8.9,0.0);
              \draw[spread] (t) -- (9.3,0.0);

              { [yshift=-1.5cm] 
                \draw[green,thick,dashed,fill=green!15] (5.25,0.0) circle [x radius=1.5cm, y radius=0.4cm];
                \draw[green,thick,dashed] (2.265,1.47) -- (3.88,-0.19);
                \draw[green,thick,dashed] (4.21,1.56) -- (6.715,0.112);

                \draw[very thick,dashed] (0,0.7) -- (10,0.7);

                \path (4.0,0.0) node (a) [input2] {}
                      (4.5,0.0) node (b) [input2] {}
                      (5.25,0) node  {$\ldots$}
                      (6.0,0.0) node (c) [input2] {}
                      (6.5,0.0) node (d) [input2] {};

                \node (w) at (0.0,0.0) {\scriptsize $w'$};
                \draw[line2] (w) -- (a);

                \node[input2_virt] (t) at (5.25,-1.0) {};

                \draw[virt] (a) -- (t);
                \draw[virt] (d) -- (t);

                \path (4.25,-0.2) node (ab) [input2_virt] {}
                      (6.25,-0.2) node (cd) [input2_virt] {};

                \draw[virt] (b) -- (ab);
                \draw[virt] (c) -- (cd);
                \draw[virt] (ab) -- (cd);

                \draw[line1] (0.8,-0.2) -- (ab);
                \coordinate (start) at (0.8,-0.2);

                \draw[line1] (0.8,-1.0) -- (t);
                \node (w) at (0.0,-1.0) {\scriptsize $w'/\hat{n}$};
                \draw[line2] (w) -- (0.8,-1.0);

                \coordinate (end) at (0.8,-1.0);
                { [>=stealth']
                  \draw[<->] (start) -- (end);
                  \node[fill=white] at ($(start)!0.5!(end)$) {\scriptsize $\log \hat{n}$};
                  \node[blue] (desc) at ($(start)!0.5!(end) + (10.2,0)$) {\scriptsize \emph{virtual} comparisons};
                  \node[blue] (desc) at ($(w) + (11,0)$) {\scriptsize last \emph{virtual} comparison};
                  \draw[line2] (t) -- (desc);
                }


                { [yshift=-1.35cm] 
                  \path (4.0,0.0) node (a) [input2] {}
                        (4.5,0.0) node (b) [input2] {}
                        (5.25,0) node  {$\ldots$}
                        (6.0,0.0) node (c) [input2] {}
                        (6.5,0.0) node (d) [input2] {};

                  \draw[line1] (0.8,0.0) -- (a);
                  \node (w)at (0.0,0.0) {\scriptsize $w'/2\hat{n}^2$};
                  \draw[line2] (w) -- (0.8,0.0);

                  \coordinate (start) at (0.8,0.0);

                  \draw[spread] (t) -- (a);
                  \draw[spread] (t) -- (b);
                  \draw[spread] (t) -- (c);
                  \draw[spread] (t) -- (d);
                  \draw[spread] (t) -- (5.75,0.0);
                  \draw[spread] (t) -- (4.75,0.0);

                  \node [below=-0.1cm of a] {\tiny $\Downarrow$};
                  \node [below=-0.1cm of b] {\tiny $\Downarrow$};
                  \node [below=-0.1cm of c] {\tiny $\Downarrow$};
                  \node [below=-0.1cm of d] {\tiny $\Downarrow$};
                  { [yshift=-0.3cm]
                    \coordinate (a) at (4.0,0.0);
                    \coordinate (b) at (4.5,0.0);
                    \coordinate (c) at (6.0,0.0);
                    \coordinate (d) at (6.5,0.0);
                    \coordinate (la) at (3.85,-0.37);
                    \coordinate (ra) at (4.15,-0.37);
                    \coordinate (lb) at (4.35,-0.37);
                    \coordinate (rb) at (4.65,-0.37);
                    \coordinate (lc) at (5.85,-0.37);
                    \coordinate (rc) at (6.15,-0.37);
                    \coordinate (ld) at (6.35,-0.37);
                    \coordinate (rd) at (6.65,-0.37);

                    \draw[line1] (0.8,-0.37) -- (la);

                    \coordinate (end) at (0.8,-0.37);
                    { [>=stealth']
                      \draw[<->] (start) -- (end);
                      \node[fill=white] at ($(start)!0.5!(end)$) {\scriptsize $\log \sqrt{\hat{n}}$};
                      \node[text width=2.6cm,text centered] at ($(start)!0.5!(end) + (10.2,0)$)
                        {\scriptsize routing through \emph{marked} elements};
                    }

                    \draw[triangle] (a) -- (la) -- (ra) -- cycle;
                    \draw[triangle] (b) -- (lb) -- (rb) -- cycle;
                    \draw[triangle] (c) -- (lc) -- (rc) -- cycle;
                    \draw[triangle] (d) -- (ld) -- (rd) -- cycle;

                    \node [below=0.3cm of a] {\tiny $\Downarrow$};
                    \node [below=0.3cm of b] {\tiny $\Downarrow$};
                    \node [below=0.3cm of c] {\tiny $\Downarrow$};
                    \node [below=0.3cm of d] {\tiny $\Downarrow$};

                    { [yshift=-0.7cm] 
                      \path (4.0,0.0) node (a) [input2_virt] {}
                            (4.5,0.0) node (b) [input2_virt] {}
                            (6.0,0.0) node (c) [input2_virt] {}
                            (6.5,0.0) node (d) [input2_virt] {};

                      \draw[line1] (0.8,0.0) -- (a);
                      \node (w) at (0.0,0.0) {\scriptsize $w'/2\hat{n}^{2.5}$};
                      \draw[line2] (w) -- (0.8,0.0);

                      \coordinate (start) at (0.8,0.0);
                      \coordinate (end) at (0.8,-1.0);
                      { [>=stealth']
                        \draw[<->] (start) -- (end);
                        \node[fill=white] (w) at ($(start)!0.5!(end)$) {\scriptsize $\log \hat{n}$};
                        \node[blue] (desc) at ($(start)!0.5!(end) + (10.2,0)$) {\scriptsize \emph{virtual} comparisons};
                      }

                      \node[input2_virt] (t) at (5.25,-1.0) {};

                      \draw[line1] (0.8,-1.0) -- (t);
                      \node (w) at (0.0,-1.0) {\scriptsize $w'/2\hat{n}^{3.5}$};
                      \node[blue] (desc) at ($(w) + (11,0)$) {\scriptsize last \emph{virtual} comparison};
                      \draw[line2] (t) -- (desc);
                      \draw[line2] (w) -- (0.8,-1.0);

                      \draw[virt] (a) -- (t);
                      \draw[virt] (d) -- (t);

                      \path (4.25,-0.2) node (ab) [input2_virt] {}
                            (6.25,-0.2) node (cd) [input2_virt] {};

                      \draw[virt] (b) -- (ab);
                      \draw[virt] (c) -- (cd);
                      \draw[virt] (a) -- (d);
                      \node[fill=white] at (5.25,0) {$\ldots$};

                      \path (3.75, -1.4) node [input2] {}
                            (4.5, -1.4) node [input2] {}
                            (5.25, -1.4) node [input2] {}
                            (6.0, -1.4) node [input2] {}
                            (6.75, -1.4) node [input2] {};
                      \node (w) at (0.0,-1.4) {\scriptsize $w'' = w'/4\hat{n}^{4.5}$};
                      \draw[line2] (w) -- (3.75,-1.4);

                      \draw[spread] (t) -- (3.75,-1.4);
                      \draw[spread] (t) -- (4.0,-1.4);
                      \draw[spread] (t) -- (4.25,-1.4);
                      \draw[spread] (t) -- (4.5,-1.4);
                      \draw[spread] (t) -- (4.75,-1.4);
                      \draw[spread] (t) -- (5.0,-1.4);
                      \draw[spread] (t) -- (5.25,-1.4);
                      \draw[spread] (t) -- (5.5,-1.4);
                      \draw[spread] (t) -- (5.75,-1.4);
                      \draw[spread] (t) -- (6.0,-1.4);
                      \draw[spread] (t) -- (6.25,-1.4);
                      \draw[spread] (t) -- (6.5,-1.4);
                      \draw[spread] (t) -- (6.75,-1.4);

                      \node[red] (desc) at (11,-1.4) {\scriptsize split into \emph{boxes}};
                      \draw[line2] (6.75,-1.4) -- (desc);
                    }
                  }
                }
              }
            }
          }
        }
      }

    \end{tikzpicture}
  \end{center}
  \caption{Algorithm and assigned weights -- first 2 levels of recursion}
  \label{fig_1}
\end{figure*}
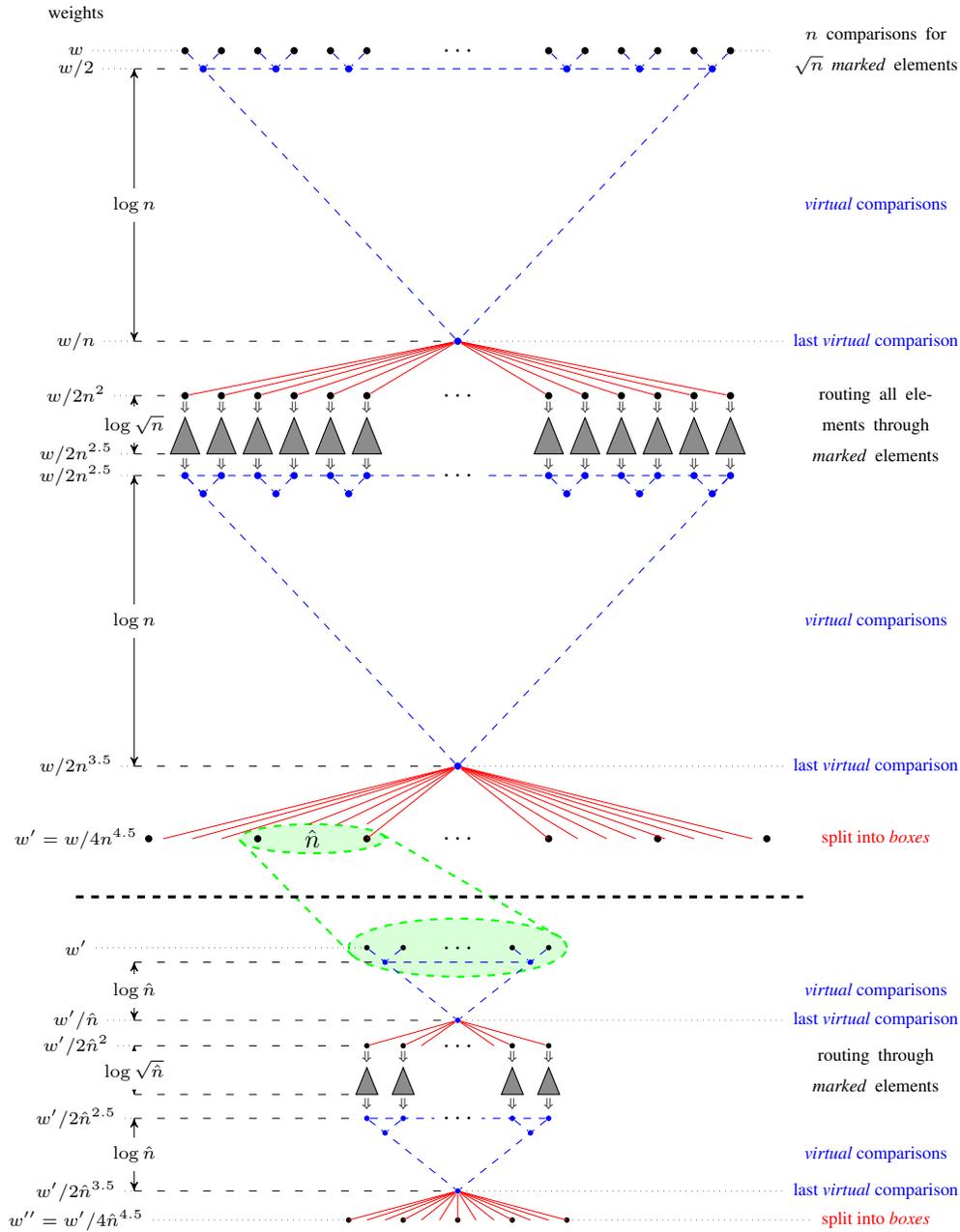 

Figure~\ref{fig_1} presents the network-like interpretation
of the algorithm, including \emph{virtual}
comparisons. It is meant to be helpful in understanding
the description of the algorithm.
A few things to keep in mind when reading it:
\begin{itemize}
  \item Nodes represent comparisons.
  \item If two comparisons are linked by an edge,
    the lower one depends on the upper one.
  \item \textbf{Black} nodes represent \emph{real} comparisons
    (ones that actually appear during algorithm execution).
  \item \textcolor{blue}{\textbf{Blue}} nodes represent
    \emph{virtual} comparisons.
  \item \textcolor{red}{\textbf{Red}} edges represent
    assigning weight to \emph{real} comparisons based on
    \emph{virtual} ones.
  \item Figure~\ref{fig_1} shows two consecutive levels of recursion
    of the algorithm; first level is shown in full;
    only one subproblem from the second level (of size $\hat{n}$) is shown.
  \item For simplicity, we assume (as we did in~\ref{algorithm_label})
    that there are exactly $n$ sorting comparisons, and exactly $n$ unmarked items
    are to be routed through the trees.
\end{itemize}

\ifx\tapmacrosareloaded\relax 
  \expandafter \else\let\tapmacrosareloaded\relax
  \immediate\write16{%
  This is TAP version 0.77 (BOP TeX macro package for setting tables)}%
\fi
\let\normalexclam!
\let\normalat@
\let\normalquotes"
\let\normalvbar|
\def\normalundscore{_}

\edef\dblquotecatcode{\the\catcode`\"}
\edef\brokenbarcatcode{\the\catcode`\|}
\edef\exclamcatcode{\the\catcode`\!}
\edef\undscorecatcode{\the\catcode`\_}
\edef\atcatcode{\the\catcode`\@}
\def\tableactive{%
  \catcode`\|13 \catcode`\@13 \catcode`\"13 \catcode`\!13\relax}

\ifx\tapspecial\unknown 
  \def\psunitsperinch{Resolution}
  \def\tapspecial#1{\special{ps:#1}}
\fi

\catcode`\_=11
\catcode`\@=11 

\def\psunits_setup{\tapspecial{%
  /in {\psunitsperinch\space mul} def
  /bp {72 div in} def
  /pt {72.27 div in} def
  /sp {65536 div pt} def
  /dd {1238 mul 1157 div pt} def
  /cc {12 mul dd} def
  /mm {25.4 div in} def
  /cm {10 mul mm} def
  /trth {\number\trth\space sp} def
  /TRTH {\number\TRTH\space sp} def
}}

\ifx\mscount\undefined_ \let\mscount\@multicnt \fi
\ifx\mscount\undefined_ \csname newcount\endcsname\mscount \fi

\tableactive

\def|{\ifmmode \vert\else \normalvbar \fi} 
\let!\normalexclam
\let@\normalat
\let"\normalquotes

\newdimen\trth \trth=.4pt
\newdimen\TRTH \TRTH=1.2pt
\newcount\align_state \align_state=0 

\def\hssf{\hskip 0pt plus 1fill minus 1fill}
\def\math#1{\relax $\relax#1\relax$}
\def\displaymath #1{\relax$\displaystyle #1\relax$}

\newtoks\everytable \everytable{\relax}
\newtoks\thistable \thistable{\relax}

\newdimen\desiredwidth
\def\width_accuracy_limit{.01pt}
\newdimen\widthaccuracy \widthaccuracy\width_accuracy_limit
\newif\iflongcalculation

\newbox\table_box

\newbox\box_tmp
\newdimen\innerht \newdimen\innerdp 
\newdimen\normht  \newdimen\normdp  
\newdimen\extraht \newdimen\extradp 

\def\defaultcellmarg{0.5em} 
\newdimen\cellmarg          

\newcount\cellnum
\newcount\maxcellnum
\newdimen\cellcorr
\newdimen\maxcellcorr
\newdimen\mincellcorr
\newcount\cellstate 

\newif\ifignorecellstate

\def\ostrut{\vrule width0mm}
\def\istrut{\ostrut height\innerht depth\innerdp\relax}

\def\defaultstrut{%
  \setbox\box_tmp\hbox{(pl}%
  \settablestrut \ht\box_tmp \dp\box_tmp \ht\box_tmp \dp\box_tmp
  \adjusttablestrut .45ex .15ex .95ex .55ex
  \setstrut \normht \normdp
}

\def\setstrut{\afterassignment\innerdp\innerht}

\def\adjuststrut{\afterassignment\adjuststrut_\advance\innerht}
\def\adjuststrut_{\advance\innerdp}

\def\settablestrut{\afterassignment\settablestrut_\normht}
\def\settablestrut_{\afterassignment\settablestrut__\normdp}
\def\settablestrut__{\afterassignment\settablestrut___\extraht}
\def\settablestrut___{\extradp}

\def\adjusttablestrut{\afterassignment\adjusttablestrut_\advance\normht}
\def\adjusttablestrut_{\afterassignment\adjusttablestrut__\advance\normdp}
\def\adjusttablestrut__{\afterassignment\adjusttablestrut___\advance\extraht}
\def\adjusttablestrut___{\advance\extradp}

\def\symcellcorr{\ignorecellstatetrue}
\def\asymcellcorr{\ignorecellstatefalse}
\symcellcorr


\newdimen\mincellmarg
\def\update_mincellmarg{%
  \dim_tmp\cellmarg \advance\dim_tmp\cellcorr
  \noexpand\ifdim\dim_tmp<\mincellmarg \global\mincellmarg\dim_tmp
  \noexpand\fi}

\def\leftcellside{%
  \global\advance\cellnum1\kern\cellmarg
  \ifignorecellstate \kern\cellcorr \update_mincellmarg
  \else {\multiply\cellcorr\cellstate \kern\cellcorr \update_mincellmarg
  }\fi
}
\def\antileftcellside{\hskip-\cellmarg
  \ifignorecellstate \kern-\cellcorr
  \else {\multiply\cellcorr\cellstate \kern-\cellcorr}\fi
}
\def\rightcellside{%
  \global\advance\cellnum1\kern\cellmarg
  \ifignorecellstate \kern\cellcorr \update_mincellmarg
  \else {\cellstate-\cellstate \advance\cellstate2
    \multiply\cellcorr\cellstate \kern\cellcorr \update_mincellmarg
  }\fi
}
\def\antirightcellside{\hskip-\cellmarg
  \ifignorecellstate \kern-\cellcorr
  \else {\cellstate-\cellstate \advance\cellstate2
    \multiply\cellcorr\cellstate \kern-\cellcorr}\fi
}

\def\beginflatcell{\vbox to0mm\bgroup\vss \hbox\bgroup \istrut}
\def\endflatcell{\egroup\vss\egroup}

\def\begindimencell{\vbox to\tablerowheight\bgroup\vss \hbox\bgroup}
\def\enddimencell{\egroup\vss\egroup}

\newif\ifv_squash
\newbox\cell_box
\def\normcellbox{\hbox\bgroup\v_squashfalse\let\v_box_or_top\vbox\cellbox_}
\def\flatcellbox{\hbox\bgroup\v_squashtrue\let\v_box_or_top\vbox\cellbox_}
\def\normcelltop{\hbox\bgroup\v_squashfalse\let\v_box_or_top\vtop\cellbox_}
\def\flatcelltop{\hbox\bgroup\v_squashtrue\let\v_box_or_top\vtop\cellbox_}
\def\cellbox_#1#{\def\cellbox_parm{#1}%
  \ifx\cellbox_parm\empty \expandafter \cellbox_a \else
  \expandafter \cellbox_p \fi}
\def\cellbox_a#1{
  \setbox\cell_box\v_box_or_top{%
    \let\\\cr
    \let@\normalat
    \let!\normalexclam
    \def\L##1{\ignorespaces ##1\unskip\hssf\null}
    \def\C##1{\hssf \ignorespaces ##1\unskip\hssf\null}
    \def\R##1{\hssf \ignorespaces ##1\unskip\null}
    \everycr{}
    \halign{%
      \istrut \ifnum\cellstate>0\hss\fi ##\ifnum\cellstate<2\hss\fi \cr
      \ostrut height\curr_distender_height
      \ignorespaces #1%
      \ifhmode\ostrut depth\curr_distender_depth\fi
      \crcr
    }
  }%
  \ifv_squash\vbox to0mm{\vss\box\cell_box\vss}\else\box\cell_box\fi
  \egroup
}
\def\cellbox_p#1{
  \setbox\cell_box\v_box_or_top{\parindent0mm \parfillskip0pt
    \def\\{\unskip\break\hskip0pt\relax\ignorespaces}
    \let@\normalat
    \let!\normalexclam
    \let~\ori_tilde 
    \let\-\ori_bsdash 
    \let\=\ori_bsequal 
    \def\L##1{\ignorespaces ##1\unskip\hfill\null}
    \def\C##1{\hfill\ignorespaces ##1\unskip\hfill\null}
    \def\R##1{\hfill\ignorespaces ##1\unskip\null}
    \hsize\cellbox_parm
    \normalbaselines \baselineskip\innerht \advance\baselineskip\innerdp
    \ifnum\cellstate>0 \leftskip0ptplus1fil \fi
    \ifnum\cellstate<2 \rightskip0ptplus1fil \fi
    \ostrut height\curr_distender_height
    \nobreak\hskip0pt\relax 
    \ignorespaces #1%
    \ifhmode\ostrut depth\curr_distender_depth\fi
  }%
  \ifv_squash \vbox to0mm{\vss\box\cell_box\vss}\else\box\cell_box\fi
  \egroup
}

\def\cellaction#1{\ifcase\align_state \def\cellaction_{#1}\or
  \def\cellaction_{#1}\or
  \def\cellaction_{\unskip&#1&}\else
  \def\cellaction_{#1}\fi\cellaction_}

\def\tablerulecmyk#1{\def\thetablerulecmyk{#1}}
\def\beginrulecmyk{\ifx\thetablerulecmyk\empty\else
  \tapspecial{gsave \thetablerulecmyk\space setcmykcolor}\fi}
\def\endrulecmyk{\ifx\thetablerulecmyk\empty\else\tapspecial{grestore}\fi}
\def\tablefullrule#1{\noalign{\beginrulecmyk \hrule height #1\endrulecmyk}}
\def\tablevrule#1{\hss\beginrulecmyk\vrule width #1\endrulecmyk\hss}
\def\tablehrule#1{\antileftcellside
  \beginrulecmyk\leaders\hrule height #1\hfill\endrulecmyk
  \antirightcellside}

\def\preambleleft{\cellstate0\relax \leftcellside \currtfont
  \begincell \ignorespaces ####\unskip\endcell \hfil \rightcellside}
\def\preamblecenter{\cellstate1\relax \leftcellside \hfil \currtfont
  \begincell \ignorespaces ####\unskip\endcell \hfil \rightcellside}
\def\preambleright{\cellstate2\relax \leftcellside \hfil \currtfont
  \begincell \ignorespaces ####\unskip\endcell \rightcellside}
\def\tableleft#1{\omit\cellstate0\relax \leftcellside \currtfont
  \begincell \ignorespaces #1\endcell \hfil \rightcellside}
\def\tablecenter#1{\omit\cellstate1\relax \leftcellside \hfil \currtfont
  \begincell \ignorespaces #1\endcell \hfil \rightcellside}
\def\tableright#1{\omit\cellstate2\relax \leftcellside \hfil \currtfont
  \begincell \ignorespaces #1\endcell \rightcellside}

\def\span_{\span\omit \advance\mscount-1\relax}
\def\usecells{\omit\relax\afterassignment\usecells_\mscount}
\def\usecells_{%
  \advance\mscount-1\multiply\mscount2\relax
  \loop\ifnum\mscount>1 \span_ \repeat
  \ifnum\mscount>0 \span \else \relax \fi}

\def\short_fix_cellcorr{%
    \cellcorr\desiredwidth \advance\cellcorr-\wd\table_box
    \def\corr_sign{}%
    \ifdim\cellcorr<0mm\def\corr_sign{-}\cellcorr-\cellcorr\fi
    \resize \cellcorr{1sp}\cellcorr{\the\maxcellnum sp}%
    \cellcorr\corr_sign\cellcorr
}

\def\long_fix_cellcorr{%
  \ifdim\widthaccuracy<\width_accuracy_limit\relax
    \widthaccuracy\width_accuracy_limit\relax 
  \fi
  \cellcorr0mm
  \dim_tmp\desiredwidth \advance\dim_tmp-\wd\table_box
  \ifdim\dim_tmp<0mm \mincellcorr.5\dim_tmp
  \else \maxcellcorr.5\dim_tmp \fi
  \loop
    \ifdim\dim_tmp<0mm \maxcellcorr\cellcorr \dim_tmp-\dim_tmp
    \else \mincellcorr\cellcorr \fi
    \ifdim\dim_tmp>\widthaccuracy
    \cellcorr\mincellcorr \advance\cellcorr\maxcellcorr
    \cellcorr.5\cellcorr
    \setbox\table_box\hbox{\nospecanchors\curr_table}%
    \dim_tmp\desiredwidth \advance\dim_tmp-\wd\table_box
  \repeat
}

\def\fix_cellcorr{%
  \iflongcalculation \long_fix_cellcorr \else \short_fix_cellcorr \fi}

\def\tap_warning{TAP warning (line \the\inputlineno)}
\def\margin_warning{%
  \ifdim\cellcorr<0mm
    \ifdim\mincellmarg<\TRTH
      \immediate\write16{%
        \tap_warning: effective cell margin = \the\mincellmarg.}%
    \fi
  \fi}

\def\width_info#1{%
  \let\ \space 
  \immediate\write16{%
    \tap_warning: Unsuccessful forcing of the desired width.}%
  \immediate\write16{%
    \ \ \ \ \ \ \ \ \ \ \ \ \ \ desired width - resulting width = \the#1}%
  \immediate\write16{%
    \ \ \ \ \ \ \ \ \ \ \ \ \ \ required accuracy = \the\widthaccuracy}}

\def\width_warning{%
  \ifdim\desiredwidth>0mm
    \begingroup
      \advance\desiredwidth-\wd\table_box 
      \ifdim\desiredwidth>\widthaccuracy \width_info\desiredwidth \fi
      \ifdim\desiredwidth<-\widthaccuracy \width_info\desiredwidth \fi
    \endgroup
  \fi}

\newif\iftapverbose
\def\tapverboseon{\tapverbosetrue}
\def\tapverboseoff{\tapverbosefalse}

\newif\ifspec_distender \spec_distenderfalse

\newdimen\tablerowheight

\def\table_setup{%
  \tableactive
  \let\ori_tilde~
  \let\ori_bsdash\-
  \let\ori_bsequal\=
  \let\ori_bsgreater\> 
  \let\ori_bsexclam\!  
  \let\ori_left\left
  \let\ori_right\right
  \let\Lslash\L 
  \let\lslash\l 
  \everymath\expandafter{\the\everymath
    \let\>\ori_bsgreater \let\!\ori_bsexclam
    \let\left\ori_left \let\right\ori_right
    \let@\normalat \let!\normalexclam
    \let"\normalquotes \let|\normalvbar
  }
  \everydisplay\expandafter{\the\everydisplay
    \let\>\ori_bsgreater \let\!\ori_bsexclam
    \let\left\ori_left \let\right\ori_right
    \let@\normalat \let!\normalexclam
    \let"\normalquotes \let|\normalvbar
  }
  \let\!\normalexclam
  \let\>\rlap
  \let\<\llap
  \let\H\hbox 
  \let\x\flatcellbox
  \let\X\normcellbox
  \let\y\flatcelltop
  \let\Y\normcelltop
  \def\P##1{{\phantom{##1}}}  
  \def\V##1{{\vphantom{##1}}} 
  \def\K{\noalign\bgroup\afterassignment\K_\dim_tmp}
  \def\K_{\kern\dim_tmp\egroup}
  \def\L{\endcell\hfill\begincell}
  \def\F##1{\noalign{\global\let\currtfont##1}} 
  \def\-{\ifcase\align_state \tablefullrule{\trth}\else
         \endcell\tablehrule{\trth}\begincell\fi}
  \def\={\ifcase\align_state \tablefullrule{\TRTH}\else
         \endcell\tablehrule{\TRTH}\begincell\fi}
  \def|{\cellaction{\tablevrule\trth}}
  \def!{\cellaction{\tablevrule\TRTH}}
  \def\br##1{%
    \global\cellnum0\relax
    \global\align_state1\relax \ignorespaces ##1\unskip
    \global\align_state2\relax &%
  }
  \def\B##1##2{%
    \def\distender_kind{##2}%
    \def\dimen_suffix{=}%
    \def\B__{\B_##1##2}%
    \ifx\distender_kind\dimen_suffix
      \afterassignment\B__ \global\tablerowheight =
    \else \global\tablerowheight=0pt\relax \expandafter \B__ \fi
  }%
  \def\B_##1##2{%
    \br{##1\spec_distendertrue\D{##2}}}
  \def\er##1{\global\align_state=3\unskip&##1\unskip \global\align_state=0
             \iftapverbose \message{*}\fi \cr}
  \def\E{\er}
  \def\D##1{
    \ifvmode\leavevmode\fi
    \ifhmode\nobreak\hskip0pt\relax\fi 
    \gdef\next_tok{}%
    \def\distender_kind{##1}%
    \def\null_suffix{0}%
    \def\flat_suffix{-}%
    \def\dimen_suffix{=}%
    \def\norm_suffix{+}%
    \edef\dn_suffix{\normalundscore}
    \def\up_suffix{^}%
    \def\updn_suffix{:}%
    \setbox\box_tmp\hbox{%
      \ostrut
      \ifx\distender_kind\null_suffix \else
      \ifx\distender_kind\norm_suffix height\normht depth\normdp \else
      \ifx\distender_kind\flat_suffix height0mm depth0mm
        \ifspec_distender
          \global\let\begincell\beginflatcell
          \global\let\endcell\endflatcell
        \fi
      \else
      \ifx\distender_kind\dimen_suffix height0mm depth0mm
        \ifspec_distender
          \global\let\begincell\begindimencell
          \global\let\endcell\enddimencell
        \else
          \immediate\write16{%
            \tap_warning: `=' following \string\D\space is equivalent to `0'}%
        \fi
      \else
      \ifx\distender_kind\up_suffix height\extraht depth\normdp \else
      \ifx\distender_kind\dn_suffix height\normht depth\extradp \else
      \ifx\distender_kind\updn_suffix height\extraht depth\extradp
      \else \gdef\next_tok{##1}
      \fi\fi\fi\fi\fi\fi\fi \relax
    }%
    \xdef\curr_distender_depth{\the\dp\box_tmp}%
    \xdef\curr_distender_height{\the\ht\box_tmp}%
    \unhbox\box_tmp
    \ifspec_distender
      \aftergroup\next_tok 
    \else \expandafter \next_tok \fi
  }
  \let@\usecells
  \def~{\leavevmode\phantom{0}}
  \defaultstrut
  \longcalculationtrue
  \cellmarg\defaultcellmarg 
  \everycr{\noalign{%
    \ifnum\cellnum>\maxcellnum \global\maxcellnum\cellnum\fi
    \global\let\begincell\defaultbegincell
    \global\let\endcell\defaultendcell
    \global\align_state0
  }}
  \let\defaultbegincell\relax
  \let\defaultendcell\relax
  \tablerulecmyk{}
  \tablebkgcmyk{0 0 0 0}
  \global\deferredtrue 
  \the\everytable\relax \the\thistable\relax}

\def\begintable_{\vbox\bgroup
  \global\let\begincell\defaultbegincell 
  \global\let\endcell\defaultendcell     
  \global\let\currtfont\relax            
  \global\maxcellnum0 \global\cellnum0 \global\mincellmarg\maxdimen
  \let\cellaction_\relax 
  \let\ori_par\par
  \let\par\empty 
  \everyvbox\expandafter{\the\everyvbox \let\par\ori_par}
}
\def\endtable_{\crcr\egroup\egroup}

\def\begintableformat #1\endtableformat{\offinterlineskip \tabskip = 0pt
  \let\left\preambleleft \let\center\preamblecenter \let\right\preambleright
  \def"{&########&}
  \edef\tbl_form{####&#1&####\cr}
  \let\left\tableleft \let\center\tablecenter \let\right\tableright
  \def"{\cellaction\relax}
  \xdef\currtfont{\the\font}
  \edef\leftcellside{\leftcellside}
  \edef\rightcellside{\rightcellside}
  \halign\bgroup\span\tbl_form}

\def\begintable{\vbox\bgroup \table_setup \midtable}
\long\def\midtable #1\endtable{%
  \def\curr_table{\begintable_#1\endtable_}\nospecanchors
  \edef\currhbadness{\the\hbadness}\edef\currvbadness{\the\vbadness}%
  \ifdim\desiredwidth>0mm \hbadness10000 \vbadness10000 \fi
  \setbox\table_box\hbox{\curr_table}%
  \ifdim\desiredwidth>0mm
    \ifdim\wd\table_box>0mm
      \fix_cellcorr
      \hbadness\currhbadness \vbadness\currvbadness
      \setbox\table_box\hbox{\curr_table}%
      \iflongcalculation \else \width_warning\fi
      \margin_warning
    \else
      \immediate\write16{%
        \tap_warning: table width = \the\wd\table_box\space
        (\desiredwidth ignored).}%
    \fi
  \fi
  \xdef\tablewidth{\the\wd\table_box}%
  \xdef\tablecellcorr{\the\cellcorr}%
  \box\table_box
  \egroup
  \thistable{\relax}}

\def\deftable#1\begintable{
  \begingroup
     \outer\def\begintable{}
     \def\table_name{#1}\tableactive
     \expandafter \expandafter \expandafter \gdef \expandafter \table_name
     \deftablecont}
\long\def\deftablecont#1\endtable{{#1}\endgroup}

\def\beginanchtable{\vbox\bgroup \psunits_setup \table_setup \midanchtable}
\long\def\midanchtable #1\endanchtable{%
  \def\curr_table{%
    \gdef\useanchors{}
    \begintable_#1\endtable_}%
  \edef\currhbadness{\the\hbadness}\edef\currvbadness{\the\vbadness}%
  \hbadness10000 \vbadness10000
  \setbox\table_box\hbox{\nospecanchors\curr_table}%
  \edef\totalcellnum{\the\maxcellnum}
  \hbox{%
    \ifdim\desiredwidth>0mm \fix_cellcorr
      \iflongcalculation \else 
        \setbox\table_box\hbox{\nospecanchors\curr_table}%
      \fi
    \fi
    \hbadness\currhbadness \vbadness\currvbadness
    \setbox\table_box\hbox{\curr_table}%
    \ifdim\desiredwidth>0mm
       \iflongcalculation \else \width_warning\fi
       \margin_warning
    \fi
    \xdef\tablewidth{\the\wd\table_box}%
    \xdef\tablecellcorr{\the\cellcorr}%
    \tbkg_patch\table_box
    \hbadness\currhbadness \vbadness\currvbadness
    \useanchors \nospecanchors \llap{\curr_table}%
  }%
  \egroup
  \thistable{\relax}}
\def\tbkg_patch#1{
 \speccastat
   {-\tablebkgcorr}{-\tablebkgcorr}
   {\castspecial{ur}\tbkgsuffix{\wd#1}{\ht#1}}%
 \speccastat
   {\tablebkgcorr}{\tablebkgcorr}
   {\castspecial{ll}\tbkgsuffix{0mm}{-\dp#1}}%
 \box#1
 \dotbkpatch 
}
\def\tbkgsuffix{TbKg.SuF}
\def\tablebkgcmyk#1{\def\thetablebkgcmyk{#1}}
\def\tablebkgcorr{0mm}
\def\dotbkpatch{\specrect\tbkgsuffix\thetablebkgcmyk}
\def\anchoffset{.5\trth}
\def\nospecanchors{%
  \def\expa_expa_##1##2##3##4{\gdef\useanchors{}}%
  \def\speccastat##1##2##3{}%
  \def\totalcellnum{0}%
}
\def\castat#1#2#3{
 \vbox to0mm{\vss\hbox to0mm{\kern#1\relax#3\hss}\kern#2}}
\def\speccastat{\castat} 
\def\castspecial#1#2#3#4{
  \speccastat{#3}{#4}
    {\tapspecial{currentpoint /y#1.#2 exch def /x#1.#2 exch def
    /xlast.#2 x#1.#2 def /ylast.#2 y#1.#2 def}}}
\def\expa_expa_#1#2#3#4{%
 \expandafter\expandafter\expandafter#1%
 \expandafter\expandafter\expandafter#2%
 \expandafter\expandafter\expandafter{%
 \expandafter#3#4}}
{\catcode`\p12 \catcode`\t12 \gdef\gobblePT#1pt{#1}}
\def\dimval#1{\expandafter\gobblePT\the#1} 
\def\gobblespace#1 #2\relax{#1}
\def\appendanchors#1{%
  \ifdeferred
  \edef\useanchors_{#1}\expa_expa_\gdef\useanchors\useanchors\useanchors_
  \else \global\deferredtrue #1\fi
}
\def\rectfill#1#2{\noalign{\appendanchors{\noexpand\specrect{#2}{#1}}}}
\def\trianfill#1#{\noalign\bgroup
   \lowercase{\edef\trian_dir{\gobblespace #1 \relax}}\trianfill_}
\def\trianfill_#1#2{\appendanchors{%
   \expandafter\noexpand\csname spectrian\trian_dir \endcsname {#2}{#1}}%
\egroup}
\newdimen\pen_wd
\def\diagstroke{\noalign\bgroup\afterassignment\diagstroke_\pen_wd}
\def\diagstroke_#1#{%
  \lowercase{\edef\diag_dir{\gobblespace #1 \relax}}\diagstroke__}
\def\diagstroke__#1#2{\appendanchors{%
  \expandafter\noexpand\csname specdiag\diag_dir \endcsname
    {#2}{#1}{\number\pen_wd}}%
  \egroup}
\def\specdiagboth#1#2#3{%
  \edef\cmyk_{#2}%
  \tapspecial{gsave newpath \ifx\cmyk_\empty\else #2 setcmykcolor \fi
    #3 sp setlinewidth 1 setlinecap
    xll.#1 yll.#1 moveto xur.#1 yur.#1 lineto stroke
    xll.#1 yur.#1 moveto xur.#1 yll.#1 lineto stroke grestore}}
\def\specdiagdown#1#2#3{%
  \edef\cmyk_{#2}%
  \tapspecial{gsave newpath \ifx\cmyk_\empty\else #2 setcmykcolor \fi
    #3 sp setlinewidth 1 setlinecap
    xll.#1 yur.#1 moveto xur.#1 yll.#1 lineto stroke grestore}}
\def\specdiagup#1#2#3{%
  \edef\cmyk_{#2}%
  \tapspecial{gsave newpath \ifx\cmyk_\empty\else #2 setcmykcolor \fi
    #3 sp setlinewidth 1 setlinecap
    xll.#1 yll.#1 moveto xur.#1 yur.#1 lineto stroke grestore}}
\def\specrect#1#2{%
  \edef\cmyk_{#2}%
  \tapspecial{gsave newpath \ifx\cmyk_\empty\else #2 setcmykcolor \fi
    xll.#1 yll.#1 moveto
    xur.#1 yll.#1 lineto xur.#1 yur.#1 lineto xll.#1 yur.#1 lineto
    closepath fill grestore}}
\def\spectrianne#1#2{%
  \edef\cmyk_{#2}%
  \tapspecial{gsave newpath \ifx\cmyk_\empty\else #2 setcmykcolor \fi
    xll.#1 yur.#1 moveto xur.#1 yll.#1 lineto xur.#1 yur.#1 lineto
    closepath fill grestore}}
\def\spectriannw#1#2{%
  \edef\cmyk_{#2}%
  \tapspecial{gsave newpath \ifx\cmyk_\empty\else #2 setcmykcolor \fi
    xll.#1 yll.#1 moveto xur.#1 yur.#1 lineto xll.#1 yur.#1 lineto
    closepath fill grestore}}
\def\spectrianse#1#2{%
  \edef\cmyk_{#2}%
  \tapspecial{gsave newpath \ifx\cmyk_\empty\else #2 setcmykcolor \fi
    xll.#1 yll.#1 moveto xur.#1 yll.#1 lineto xur.#1 yur.#1 lineto
    closepath fill grestore}}
\def\spectriansw#1#2{%
  \edef\cmyk_{#2}%
  \tapspecial{gsave newpath \ifx\cmyk_\empty\else #2 setcmykcolor \fi
    xll.#1 yll.#1 moveto xur.#1 yll.#1 lineto xll.#1 yur.#1 lineto
    closepath fill grestore}}

\newcount\col_num
\newtoks\col_tok

\newif\ifdeferred

\def\ontext{\noalign{\global\deferredfalse}}
\def\clearafters{\noalign{\gdef\aftersuffix{}\gdef\afteranchor{}}}

\def\uranchor{\clearafters\corneranchor{ur}} 
\def\llanchor{\clearafters\corneranchor{ll}} 
\def\genanchor#1#2#3#4{
  \noalign{\gdef\aftersuffix{#3}\gdef\afteranchor{#4}}%
  \corneranchor{#1#2}%
}

\def\corneranchor#1{
  \noalign\bgroup \def\next_{\corneranchor_#1}\afterassignment\next_
  \global\col_num}
\def\corneranchor_#1#2{
     \if#1n%
       \ifnum\col_num<1
         \global\col_num\if#2r\totalcellnum \global\divide\col_num2
       \else1\fi
     \fi\fi
     \global\advance\col_num-1\relax
     \ifnum\col_num<0 \global\let\corneranchor__\cornerfullanchor
     \else \global\let\corneranchor__\cornercellanchor \fi
  \egroup \corneranchor__ #1#2}

\def\cornerfullanchor#1#2#3{
  \br{\speccastat{\if#2r\wd\table_box\else0mm\fi}{0mm}
        {\castspecial {#1#2}{#3\aftersuffix}
          {\if#1n0mm\else\if#2r-\fi\anchoffset\fi}
          {\if#1n0mm\else\if#1l-\fi\anchoffset\fi}}}%
  \er{}\noalign{\xdef\lastsuffix{#3}\afteranchor}\clearafters}

\def\cornercellanchor#1#2#3{
  \br{}
  \col_tok{}%
  \loop \ifnum\col_num>0
    \global\advance\col_num-1 \col_tok\expandafter{\the\col_tok"}%
  \repeat
  \the\col_tok
  \if#2r\hskip0ptplus1filll\fi
  \speccastat
    {\if#2l-\fi \ifignorecellstate 1\else
     \ifcase\cellstate \if#2l0\else2\fi\or 1\or \if#2l2\else0\fi\fi\fi
     \cellcorr}{0mm}
       {\speccastat {\if#2l-\fi\cellmarg}{0mm}
         {\castspecial {#1#2}{#3\aftersuffix}
           {\if#1n0mm\else\if#2l-\fi\anchoffset\fi}
           {\if#1n0mm\else\if#1l-\fi\anchoffset\fi}}%
  }%
  \if#2l\hskip0ptplus1filll\null\fi
  \er{}\noalign{\xdef\lastsuffix{#3}\afteranchor}\clearafters}
\def\bpsinsertion#1#2#3{\genanchor #1#2{#3}{}}
\def\epsinsertion#1#2#3#4{%
  \genanchor#2#3{#4}{\appendanchors{\tapspecial{#1}}}%
}

\def\bgenstroke#1#2{\bpsinsertion #1#2{.bs}}
\def\egenstroke#1#2#3{%
  \epsinsertion {gsave newpath
     xlast.\lastsuffix.bs ylast.\lastsuffix.bs moveto
     xlast.\lastsuffix.es ylast.\lastsuffix.es lineto 
     trth setlinewidth 
     #1 
     stroke grestore
  }#2#3{.es}%
}

\def\bgenrectangle#1#2{\bpsinsertion #1#2{.br}}
\def\egenrectangle#1#2#3{%
  \epsinsertion{gsave
     xlast.\lastsuffix.br ylast.\lastsuffix.br
     xlast.\lastsuffix.er xlast.\lastsuffix.br sub
     ylast.\lastsuffix.er ylast.\lastsuffix.br sub
     #1 
     rectfill grestore
  }#2#3{.er}%
}

\def\bstroke{\bgenstroke}
\def\estroke{
  \noalign\bgroup \afterassignment\estroke_ \global\pen_wd}
\def\estroke_#1#2{\egroup \egenstroke{\number\pen_wd\space sp setlinewidth
   mark #1 counttomark 4 eq {setcmykcolor} if cleartomark
   [[#2] {pt} forall] 0 setdash}}

\def\brectangle{\bgenrectangle}
\def\erectangle#1{
  \egenrectangle{mark #1 counttomark 4 eq {setcmykcolor} if cleartomark}}

%
%
\def\resize
  #1
  #2
  #3
  #4
  {%
  \dim_r#2\relax \dim_x#3\relax \dim_t#4\relax
  \dim_tmp=\dim_r \divide\dim_tmp\dim_t
  \dim_y=\dim_x \multiply\dim_y\dim_tmp
  \multiply\dim_tmp\dim_t \advance\dim_r-\dim_tmp
  \dim_tmp=\dim_x
  \loop \advance\dim_r\dim_r \divide\dim_tmp 2
  \ifnum\dim_tmp>0
    \ifnum\dim_r<\dim_t\else
      \advance\dim_r-\dim_t \advance\dim_y\dim_tmp \fi
  \repeat
  #1\dim_y\relax}
\newdimen\dim_x    
\newdimen\dim_y    
\newdimen\dim_t    
\newdimen\dim_r    
\newdimen\dim_tmp 

\catcode`\"\dblquotecatcode
\catcode`\|\brokenbarcatcode
\catcode`\!\exclamcatcode
\catcode`\_=\undscorecatcode
\catcode`\@=\atcatcode

\subsection{Experimental Results} \label{sec_impl}

\begin{table*}[t] 
\small
\desiredwidth=\textwidth
\begintable
\begintableformat \right " &\right \endtableformat
\=
\B!^ ! @5\center{\texttt{boxsort}-based} ! @5\center{\texttt{quicksort}-based} \E!
\B!- \center{input} ! @5\= ! @5\= \E!
\B!^ ! @3\center{time [ms]} | \center{avg} | \center{avg} ! @3\center{time [ms]} | \center{avg} | \center{avg} \E!
\B!- \center{size ($n$)} ! @3\- | | ! @3\- | | \E!
\B!^ ! avg | min | max | \center{comp} | \center{resol} ! avg | min | max | \center{comp} | \center{resol} \E!
\-
\B!^ 10 ! 0.8 | 0 | 2 | 11.4 | 7.2 ! 0.6 | 0 | 2 | 11.7 | 6.7 \E!
\-
\B!^ 100 ! 40.5 | 33 | 48 | 43.6 | 14.3 ! 37.9 | 30 | 56 | 63.1 | 15.4 \E!
\-
\B!^ 500 ! 323.1 | 282 | 365 | 69.2 | 20.0 ! 322.8 | 271 | 410 | 129.2 | 20.5 \E!
\-
\B!^ 1000 ! 736.4 | 644 | 809 | 79.3 | 22.2 ! 753.2 | 610 | 919 | 163.9 | 23.5 \E!
\-
\B!^ 2000 ! 1674.1 | 1474 | 1852 | 89.4 | 24.6 ! 1720.6 | 1475 | 2047 | 206.9 | 26.4 \E!
\-
\B!^ 5000 ! 4788.9 | 4353 | 5386 | 102.7 | 27.5 ! 4986.3 | 4322 | 5838 | 268.0 | 28.4 \E!
\-
\B!^ 8000 ! 8194.6 | 7366 | 9218 | 111.1 | 30.2 ! 8648.8 | 7516 | 10252 | 303.7 | 31.1 \E!
\-
\B!^ 10000 ! 10480.5 | 9189 | 11843 | 113.3 | 30.4 ! 10986.5 | 9773 | 13210 | 323.8 | 29.7 \E!
\-
\B!^ 12000 ! 13081.8 | 11995 | 14688 | 116.7 | 31.8 ! 13855.2 | 12272 | 16810 | 334.1 | 32.4 \E!
\-
\B!^ 14000 ! 15532.0 | 14127 | 17382 | 118.7 | 32.4 ! 16233.5 | 14695 | 18813 | 346.4 | 32.7 \E!
\-
\B!^ 16000 ! 17926.0 | 15925 | 19975 | 120.4 | 31.8 ! 18914.1 | 16779 | 21375 | 360.8 | 32.5 \E!
\-
\B!^ 18000 ! 20618.1 | 18175 | 23183 | 121.3 | 31.8 ! 21764.0 | 19081 | 24505 | 369.1 | 33.1 \E!
\-
\B!^ 20000 ! 23413.6 | 21516 | 26015 | 124.2 | 32.6 ! 24406.6 | 20939 | 28470 | 377.4 | 34.2 \E!
\=
\endtable
\caption{Test results for the median-of-lines problem;
avg comp -- average number of median comparison resolutions;
avg resol -- average number of median comparisons actually resolved using
$\mathcal{C}$}
\label{table_mlines}
\end{table*} 

\begin{table*}[t] 
\small
\desiredwidth=\textwidth
\begintable
\begintableformat \right " &\right \endtableformat
\=
\B!^ ! @5\center{\texttt{boxsort}-based} ! @5\center{\texttt{quicksort}-based} \E!
\B!- \center{input} ! @5\= ! @5\= \E!
\B!^ ! @3\center{time [s]} | \center{avg} | \center{avg} ! @3\center{time [s]} | \center{avg} | \center{avg} \E!
\B!- \center{size ($n$)} ! @3\- | | ! @3\- | | \E!
\B!^ ! avg | min | max | \center{comp} | \center{resol} ! avg | min | max | \center{comp} | \center{resol} \E!
\-
\B!^ 10 ! 0.003 | 0.001 | 0.004 | 33.3 | 11.0 ! 0.003 | 0.001 | 0.005 | 49.0 | 10.4 \E!
\-
\B!^ 100 ! 0.027 | 0.012 | 0.037 | 69.9 | 17.1 ! 0.023 | 0.010 | 0.037 | 141.5 | 17.2 \E!
\-
\B!^ 500 ! 0.283 | 0.133 | 0.428 | 97.4 | 23.1 ! 0.258 | 0.116 | 0.466 | 238.8 | 21.7 \E!
\-
\B!^ 1000 ! 1.167 | 0.515 | 2.352 | 106.4 | 22.8 ! 1.104 | 0.511 | 1.906 | 281.4 | 23.0 \E!
\-
\B!^ 1500 ! 2.355 | 0.712 | 4.466 | 113.3 | 23.5 ! 2.122 | 0.813 | 4.914 | 325.8 | 19.1 \E!
\-
\B!^ 2000 ! 4.628 | 2.002 | 9.222 | 116.5 | 20.2 ! 4.335 | 1.475 | 8.427 | 344.0 | 22.8 \E!
\-
\B!^ 3000 ! 11.619 | 3.151 | 22.068 | 121.6 | 21.2 ! 9.466 | 2.609 | 17.662 | 365.4 | 18.8 \E!
\-
\B!^ 4000 ! 18.461 | 7.910 | 36.385 | 124.9 | 19.6 ! 17.301 | 8.274 | 30.196 | 409.6 | 21.7 \E!
\-
\B!^ 5000 ! 30.618 | 9.844 | 56.330 | 124.7 | 18.2 ! 25.990 | 9.430 | 49.221 | 425.9 | 20.8 \E!
\-
\B!^ 6000 ! 42.180 | 14.767 | 73.373 | 128.0 | 21.6 ! 35.607 | 18.670 | 63.048 | 433.6 | 17.5 \E!
\-
\B!^ 7000 ! 58.630 | 15.824 | 137.950 | 132.3 | 19.0 ! 48.270 | 12.282 | 113.416 | 447.9 | 19.9 \E!
\-
\B!^ 8000 ! 77.023 | 23.396 | 128.263 | 133.6 | 18.6 ! 66.104 | 25.622 | 123.445 | 459.2 | 21.6 \E!
\-
\B!^ 9000 ! 95.203 | 44.364 | 174.155 | 132.6 | 18.7 ! 73.835 | 28.754 | 131.519 | 464.8 | 17.7 \E!
\-
\B!^ 10000 ! 105.006 | 37.494 | 215.619 | 134.6 | 19.8 ! 96.872 | 38.702 | 187.137 | 474.9 | 20.4 \E!
\-
\B!^ 11000 ! 136.610 | 30.205 | 254.993 | 136.3 | 18.1 ! 122.691 | 47.178 | 241.167 | 488.9 | 17.5 \E!
\-
\B!^ 12000 ! 168.988 | 63.704 | 335.881 | 138.3 | 19.8 ! 139.351 | 68.135 | 239.619 | 496.0 | 17.8 \E!
\-
\B!^ 13000 ! 192.149 | 59.290 | 399.326 | 135.5 | 21.2 ! 162.257 | 49.541 | 376.480 | 505.2 | 19.3 \E!
\-
\B!^ 14000 ! 209.349 | 70.579 | 461.083 | 138.8 | 18.4 ! 178.618 | 63.412 | 372.113 | 524.4 | 17.8 \E!
\-
\B!^ 15000 ! 243.451 | 73.700 | 405.206 | 139.7 | 16.5 ! 208.025 | 66.146 | 409.933 | 518.0 | 18.2 \E!
\=
\endtable
\caption{Test results for the point labeling problem;
avg comp -- average number of median comparison resolutions;
avg resol -- average number of median comparisons actually resolved using
$\mathcal{C}$}
\label{table_labeling}
\end{table*} 

\begin{table*}[t] 
\small
\desiredwidth=\textwidth
\begintable
\begintableformat \right " &\right \endtableformat
\=
\B!^ ! @5\center{\texttt{boxsort}-based} ! @5\center{\texttt{quicksort}-based} \E!
\B!- \center{input} ! @5\= ! @5\= \E!
\B!^ ! @3\center{time [s]} | \center{avg} | \center{avg} ! @3\center{time [s]} | \center{avg} | \center{avg} \E!
\B!- \center{size ($n$)} ! @3\- | | ! @3\- | | \E!
\B!^ ! avg | min | max | \center{comp} | \center{resol} ! avg | min | max | \center{comp} | \center{resol} \E!
\-
\B!^ 4 ! 0.038 | 0.014 | 0.070 | 77.1 | 18.8 ! 0.032 | 0.012 | 0.057 | 125.8 | 16.7 \E!
\-
\B!^ 8 ! 0.465 | 0.169 | 0.838 | 122.5 | 26.4 ! 0.393 | 0.152 | 0.686 | 273.1 | 22.8 \E!
\-
\B!^ 12 ! 1.853 | 0.875 | 3.049 | 147.8 | 29.7 ! 1.692 | 0.771 | 2.777 | 367.5 | 27.1 \E!
\-
\B!^ 16 ! 4.861 | 1.944 | 7.269 | 163.3 | 33.6 ! 4.473 | 1.702 | 6.866 | 433.8 | 29.5 \E!
\-
\B!^ 20 ! 9.552 | 5.193 | 13.527 | 173.9 | 35.3 ! 9.064 | 4.469 | 11.916 | 438.7 | 32.0 \E!
\-
\B!^ 24 ! 16.145 | 8.633 | 21.313 | 181.9 | 37.4 ! 15.415 | 7.717 | 20.654 | 541.3 | 33.5 \E!
\-
\B!^ 28 ! 23.892 | 10.435 | 32.191 | 194.0 | 38.5 ! 23.817 | 10.606 | 31.826 | 579.1 | 36.2 \E!
\-
\B!^ 32 ! 35.894 | 22.690 | 48.253 | 205.7 | 42.0 ! 35.064 | 15.580 | 46.285 | 620.1 | 38.1 \E!
\-
\B!^ 36 ! 46.377 | 29.302 | 56.620 | 211.6 | 38.7 ! 47.254 | 27.589 | 63.524 | 658.3 | 37.2 \E!
\-
\B!^ 40 ! 62.825 | 39.393 | 78.046 | 217.1 | 42.3 ! 63.635 | 37.847 | 77.617 | 682.3 | 39.3 \E!
\-
\B!^ 44 ! 79.575 | 54.912 | 100.601 | 214.9 | 41.9 ! 80.520 | 54.202 | 101.100 | 700.9 | 40.3 \E!
\-
\B!^ 48 ! 99.633 | 67.896 | 123.148 | 220.8 | 43.1 ! 100.102 | 62.814 | 124.495 | 726.1 | 39.4 \E!
\-
\B!^ 52 ! 122.606 | 77.710 | 156.918 | 215.4 | 44.2 ! 122.633 | 78.534 | 150.289 | 738.7 | 39.1 \E!
\-
\B!^ 56 ! 148.018 | 98.792 | 191.290 | 225.7 | 43.9 ! 149.623 | 94.775 | 186.019 | 766.2 | 40.1 \E!
\=
\endtable
\caption{Test results for the graph matching problem;
avg comp -- average number of median comparison resolutions;
avg resol -- average number of median comparisons actually resolved using
$\mathcal{C}$}
\label{table_g2g}
\end{table*} 

In order to measure the performance of our method
we implemented two general parametric search frameworks:
our \verb\boxsort\-based method and the
\verb\quicksort\-based algorithm
of~van~Oostrum and Veltkamp~\cite{DBLP:journals/comgeo/OostrumV04}.
We then implemented three known algorithms that utilize
parametric search,
and compared the running times of both approaches
when used in these algorithms (answers were obviously the same).
We start with the description of the problems and implemented algorithms.

\subsubsection{Implemented Algorithms}

\begin{description}
  \item{\bf Point labeling}.
    The input to the \emph{point labeling}
    problem is a set of points in the plane (\emph{objects}),
    and sets of rectangles (\emph{labels}), one set per point
    (their elements are called \emph{candidate labels}).
    A \emph{feasible} solution is one where each \emph{object} is
    assigned a \emph{label} (from its \emph{candidate labels}),
    the \emph{labels} are drawn ``near'' their \emph{objects}, and
    they do not overlap.
    The goal is to find a label placement and the largest
    \emph{scaling factor} $\sigma > 0$,
    so that the solution is still \emph{feasible} when
    dimensions of \emph{labels} are multiplied by $\sigma$.
    The problem is motivated by applications in Geographic
    Information Systems, where the goal is to label objects
    with largest non-overlapping labels (to improve readability).

    Koike~{\it et al.}~\cite{knntw-lprvs-02} gave an algorithm
    for the special case when ``near'' is defined as
    ``each \emph{label} has a special \emph{pinning point},
    and \emph{label} has to be placed so that its \emph{pinning point}
    is exactly over the \emph{object}''. They show
    that if the \emph{candidate labels} meet an additional requirement
    on the relative placement of their \emph{pinning points},
    it can be decided if \emph{candidate labels} scaled by a
    given $\sigma$ have a \emph{feasible} placement,
    in time $O(n\log n)$, by a variant of the plane sweep method,
    where $n$ is the number of \emph{objects}.

    They use this algorithm as the decision algorithm, $\mathcal{C}$,
    for parametric search. Algorithm $\mathcal{A}$ from the
    parametric search setting is defined as sorting the (symbolic) coordinates
    of labels for $\sigma^*$ -- the optimal value of the \emph{scaling factor}.
    Therefore, their parametric search algorithm for finding
    $\sigma^*$ works in $O(n\log^2n)$ time.

    An example is shown in Figure~\ref{point_labeling_example}.

  \item{\bf Matching planar drawings of graphs}.
    This problem was defined by Alt~{\it et al.}~\cite{Alt2003262},
    and is a generalization of the problem of computing the
    Fr{\'e}chet distance between two polygonal curves to computing
    a Fr{\'e}chet-like distance between two graphs, $G$ and $H$,
    embedded on a plane.

    A common explanation of the Fr{\'e}chet distance between two curves, $g$ and $h$, is
    the following: ``A man is walking his dog.
    What is the smallest length of the leash that allows the dog to traverse
    entire $h$, while simultaneously the man traverses entire $g$?''
    In the man-dog setting, the generalized problem for graphs $G$ and $H$
    is stated the following way: ``What is the smallest length of the leash that
    allows the dog to traverse the entire graph $H$, while simultaneously
    the man traverses \emph{some} part of $G$?''
    This intuition translates into a parametric transformation that can,
    for example,
    be used to produce a minimum-deformation morphing of one straight-line graph
    drawing into another.

    Alt {\it et al.} gave an $O(pq\log(pq))$ algorithm for deciding if
    a given $\lambda$, leash length, is sufficient, where $p$, $q$ are
    respective numbers of edges in $G$ and $H$.
    They construct a graph that relates possible movements of the man along
    edges of $G$ and the movements of the dog along edges of $H$
    with a leash
    of length $\lambda$. Then they perform a search in this graph to find a
    path that describes contiguous movement of man and dog that covers
    the entire $H$. It is used as the algorithm $\mathcal{C}$ from the
    generic parametric search setting.
    $\mathcal{A}$ is again defined as sorting items that describe
    the features of the aforementioned graph for optimal length
    of the leash, $\lambda^*$ (see~\cite{Alt2003262} for details).
    It yields an $O(pq\log^2(pq))$-time parametric search
    algorithm for finding $\lambda^*$.

    An example is shown in Figure~\ref{H_G_graphs}.

  \item{\bf Median of lines}.
    This ``toy problem'' was used by Megiddo
    to introduce and explain the parametric search technique
    in the original paper on the subject~\cite{DBLP:journals/jacm/Megiddo83}.
    The input consists of $n$ lines on a plane
    and the goal is to find a line whose intersection with the $x$-axis
    has the same number of lines directly above and below it.
    Megiddo gave a simple
    (although non-optimal) algorithm that utilizes
    parametric search and works in time $O(n\log n)$.

    An example is shown in Figure~\ref{median_of_lines_example}.
\end{description}

\subsubsection{Testing}

Algorithms for the above problems were implemented in \verb\C++\,
and tested on random inputs on two PC's running on Linux:

\begin{itemize}
  \item median of lines:  800 MHz CPU, 4~GB memory
  \item point labeling and matching planar
    drawings of graphs: 2.50 GHz CPU, 4~GB memory.
\end{itemize}

For the point labeling problem, inputs were
$n$ \emph{objects} with integer coordinates randomly chosen from
0-$10^6$, with 6 \emph{candidate labels} each,
pinned at their lower left corners.
For the graph matching problem, inputs were two graphs
with $n$ vertices and $O(n)$ edges each. Vertices were placed
at points with rational coordinates, with numerator
and denominator drawn independently from 1-$10^4$
(the complexity of the algorithm being $O(n^2\log^2{n})$ in this case).
For the median of lines problem, inputs were $n$ lines
defined by the $y=ax+b$ equation, with randomly
chosen rational $a$ and $b$ values.
We ran the tests on numerous values of $n$, doing 100 tests
for each value of $n$. The results are presented in
Tables~\ref{table_mlines}, \ref{table_labeling} and~\ref{table_g2g}.

Recall that comparisons are of the form $\lambda < \lambda^*$.
If we have previously determined
that $\lambda_1 < \lambda^*$ ($\lambda_1 > \lambda^*$), and
comparison $\lambda_2 < \lambda^*$ is the new median comparison we
want resolved, there is no need to invoke $\mathcal{C}$ when
$\lambda_2 < \lambda_1$ ($\lambda_2 > \lambda_1$),
as the answer is obvious.
Tables~\ref{table_mlines}, \ref{table_labeling} and~\ref{table_g2g},
show both the number
of times the algorithms needed to call $\mathcal{C}$, and the
number of total median comparisons the algorithms wanted resolved
(the difference being the number of immediately resolved median comparisons).

\begin{figure*}[ht] 
  \begin{center}
    \begin{tikzpicture}
      [
        black_node/.style={fill,circle,inner sep=1.5pt},
      ]

      \begin{scope}[scale=.6]
        \node[black_node] (a) at (0, 0) {};
        \node[black_node] (b) at (4.5, 0) {};
        \node[black_node] (c) at (4.5, 1.5) {};
        \node[black_node] (d) at (1.5, 1.5) {};
        \node[black_node] (e) at (2.5, 4) {};
        \node[black_node] (f) at (0, 5) {};
        \node[black_node] (g) at (1, 5.5) {};
        \node[black_node] (h) at (5.5, 5) {};

        \node[below=0.1 of a] {(0, 0)};
        \node[below=0.1 of b] {(9, 0)};
        \node[below=0.1 of c] {(9, 3)};
        \node[below=0.1 of d] {(3, 3)};
        \node[below=0.1 of e] {(5, 8)};
        \node[below=0.1 of f] {(0, 10)};
        \node[below=0.1 of g] {(2, 11)};
        \node[below=0.1 of h] {(11, 10)};

        \node at (8, -2) {a)};

        \begin{scope}[yshift=-380]
          \node[black_node] (a) at (0, 0) {};
          \node[black_node] (b) at (4.5, 0) {};
          \node[black_node] (c) at (4.5, 1.5) {};
          \node[black_node] (d) at (1.5, 1.5) {};
          \node[black_node] (e) at (2.5, 4) {};
          \node[black_node] (f) at (0, 5) {};
          \node[black_node] (g) at (1, 5.5) {};
          \node[black_node] (h) at (5.5, 5) {};

          \draw (a) rectangle (4.46,1.44);
          \draw (d) rectangle (4.46, 3.71);
          \draw (g) rectangle (5.44, 7);

          \draw (f) -- (0.75, 5);
          \draw (f) -- (0, 8.75);
          \draw (0.75, 5) -- (0.75, 8.75);
          \draw[decorate, decoration=zigzag] (0, 8.75) -- (0.75, 8.75);
          \draw[decorate, decoration=zigzag] (0, 9) -- (0.75, 9);
          \draw (0, 9) -- (0, 10);
          \draw (0.75, 9) -- (0.75, 10);
          \draw (0, 10) -- (0.75, 10);

          \draw (b) -- (4.5, 0.75);
          \draw (b) -- (8.75, 0);
          \draw (4.5, 0.75) -- (8.75, 0.75);
          \draw (9, 0) -- (12, 0);
          \draw (9, 0.75) -- (12, 0.75);
          \draw (12, 0) -- (12, 0.75);
          \draw[decorate, decoration=zigzag] (8.75, 0) -- (8.75, 0.75);
          \draw[decorate, decoration=zigzag] (9, 0) -- (9, 0.75);

          \draw (c) -- (4.5, 2.25);
          \draw (c) -- (8.75, 1.5);
          \draw (4.5, 2.25) -- (8.75, 2.25);
          \draw (9, 1.5) -- (12, 1.5);
          \draw (9, 2.25) -- (12, 2.25);
          \draw (12, 1.5) -- (12, 2.25);
          \draw[decorate, decoration=zigzag] (8.75, 1.5) -- (8.75, 2.25);
          \draw[decorate, decoration=zigzag] (9, 1.5) -- (9, 2.25);

          \draw (e) -- (2.5, 4.75);
          \draw (e) -- (8.75, 4);
          \draw (2.5, 4.75) -- (8.75, 4.75);
          \draw (9, 4) -- (10, 4);
          \draw (9, 4.75) -- (10, 4.75);
          \draw (10, 4) -- (10, 4.75);
          \draw[decorate, decoration=zigzag] (8.75, 4) -- (8.75, 4.75);
          \draw[decorate, decoration=zigzag] (9, 4) -- (9, 4.75);

          \draw (h) -- (5.5, 5.75);
          \draw (h) -- (8.75, 5);
          \draw (5.5, 5.75) -- (8.75, 5.75);
          \draw (9, 5) -- (13, 5);
          \draw (9, 5.75) -- (13, 5.75);
          \draw (13, 5) -- (13, 5.75);
          \draw[decorate, decoration=zigzag] (8.75, 5) -- (8.75, 5.75);
          \draw[decorate, decoration=zigzag] (9,5) -- (9, 5.75);

          \node at (2.25, 0.75) {$3\times9$};
          \node at (6, 0.35) {$1.5\times18$};
          \node at (6, 1.85) {$1.5\times18$};
          \node at (3, 2.625) {$6\times4.5$};
          \node at (4, 4.35) {$1.5\times18$};
          \node at (6.75, 5.35) {$1.5\times18$};
          \node at (3.25, 6.25) {$3\times9$};
          \node[rotate=-270] at (0.35, 6.25) {$18\times1.5$};

          \node at (8, -1) {b)};
        \end{scope}
      \end{scope}
    \end{tikzpicture}
  \end{center}
  \caption{Illustration of the performance of the point labeling algorithm.\newline
          a) \emph{objects} for the point labeling problem;
           \emph{candidate labels} for each point are rectangles
                 of proportions (width$\times$height): $1\times12$, $2\times6$,
                 $3\times4$, $4\times3$, $6\times2$ and $12\times1$.
                  \newline
                 b) \emph{feasible} placement of labels with
                 scaling factor $\sigma = 1.5$.}
  \label{point_labeling_example}
\end{figure*} 

\begin{figure*}[p] 
  \makebox[\textwidth][c]{
    \begin{tikzpicture}
      [
        blue_node/.style={fill=blue,circle,inner sep=1.5pt},
        black_node/.style={fill,circle,inner sep=1.5pt},
        empty/.style={fill=black!35},
        red_e/.style={red,very thick}
      ]

      \begin{scope}[scale=0.5]
        \node[black_node] (G_0) at (0, 2) {};
        \node[black_node] (G_1) at (2, 9) {};
        \node[black_node] (G_2) at (9, 9) {};
        \node[black_node] (G_3) at (4, -2) {};

        \draw (G_0) -- (G_1) -- (G_2) -- (G_3) -- (G_0);
        \draw[red, very thick] ($(G_0)!.55!(G_1)$) -- (G_1);
        \draw[red, very thick] (G_1) -- (G_2);
        \draw[red, very thick] (G_2) -- ($(G_2)!.74!(G_3)$);

        \node[left=0.1 of G_0] (G_0_name) {$w_0$};
        \node[below=0.1 of G_0_name] {(0, 2)};

        \node[above=0.1 of G_1] {(2, 9)};
        \node[left=0.1 of G_1] {$w_1$};

        \node[above=0.1 of G_2] {(9, 9)};
        \node[right=0.1 of G_2] {$w_2$};

        \node[right=0.1 of G_3] (G_3_name) {$w_3$};
        \node[below=0.1 of G_3_name] {(4, $-2$)};
        \node at (1.2,5) {$x$};
        \node at (5.5, 8.6) {$y$};
        \node at (6,1.5) {$z$};
        \node at (1.5,0) {$t$};

        \node[blue_node] (H_0) at (1, 7) {};
        \node[blue_node] (H_1) at (5, 10) {};
        \node[blue_node] (H_2) at (9, 8) {};
        \node[blue_node] (H_3) at (5, 2) {};

        \node[blue, left=0.1 of H_0] (H_0_name) {$v_0$};
        \node[blue, below=0.1 of H_0_name] {(1, 7)};

        \node[blue, below=0.1 of H_1]  {$v_1$};
        \node[blue, above=0.1 of H_1] {(5, 10)};

        \node[blue,right=0.1 of H_2] (H_2_name) {$v_2$};
        \node[blue,below=0.1 of H_2_name] {(9, 8)};

        \node[blue, left=0.1 of H_3] (H_3_name) {$v_3$};
        \node[blue, below=0.1 of H_3_name] {(5, 2)};

        \node[blue] at (3,8) {$a$};
        \node[blue] at (7,9.5) {$b$};
        \node[blue] at (5.9,4) {$c$};

        \draw[blue] (H_0) -- (H_1) -- (H_2) -- (H_3);
      \end{scope}

      \begin{scope}[xshift=180,yshift=140,scale=.6]
        \begin{scope}
          \clip (-0.1,0) rectangle (2.1,2);
          \draw[empty] (0,0) rectangle (2,2);
          \fill[white] (0,1.263) circle[x radius=0.25,y radius=1.11,rotate=-50];
          \node at (1,1) {\large $C_{a,x}$};
        \end{scope}

        \begin{scope}
          \clip (3.99,0) rectangle (6.1,2);
          \draw[empty] (4,0) rectangle (6,2);
          \fill[white] (5.65,0.4045) circle[x radius=0.2,y radius=0.5,rotate=-40];
          \node at (5,1) {\large $C_{a,y}$};
        \end{scope}

        \draw[empty] (8,0) rectangle (10,2);
        \node at (9,1) {\large $C_{a,z}$};

        \draw[empty] (0,-4) rectangle (2,-2);
        \node at (1,-3) {\large $C_{a,t}$};

        \draw[empty] (4,-4) rectangle (6,-2);
        \node at (5,-3) {\large $C_{b,x}$};
        \begin{scope}
          \clip (7.99,-4) rectangle (10.1,-2);
          \draw[empty] (8,-4) rectangle (10,-2);
          \fill[white] (9.55, -2.073) circle[x radius=0.2,y radius=1.9,rotate=-55];
          \node at (9,-3) {\large $C_{b,y}$};
        \end{scope}

        \begin{scope}
          \clip (-0.1,-8) rectangle (2.1,-6);
          \draw[empty] (0,-8) rectangle (2,-6);
          \fill[white] (2,-8.1) circle[x radius=0.6,y radius=1, rotate=-70];
          \node at (1,-7) {\large $C_{b,z}$};
        \end{scope}

        \draw[empty] (4,-8) rectangle (6,-6);
        \node at (5,-7) {\large $C_{b,t}$};

        \draw[empty] (8,-8) rectangle (10,-6);
        \node at (9,-7) {\large $C_{c,x}$};

        \draw[empty] (0,-12) rectangle (2,-10);
        \fill[white] (0,-10) circle (0.01);
        \node at (1,-11) {\large $C_{c,y}$};

        \begin{scope}
          \clip (3.99,-12) rectangle (6.1,-10);
          \draw[empty] (4,-12) rectangle (6,-10);
          \fill[white] (5,-11.175) circle[x radius=0.3,y radius=1.8, rotate=-60];
          \node at (5,-10.5) {\large $C_{c,z}$};
        \end{scope}
        \draw[empty] (8,-12) rectangle (10,-10);
        \node at (9,-11) {\large $C_{c,t}$};

        \draw[red_e] (0.82,2) .. controls (0.82,2.6) and (1.2,2.6) .. (1.6,2.6);
        \draw[red_e] (1.6,2.6) .. controls (2.8, 2.6) and (3.0, -0.5) .. (4.0, -0.5);
        \draw[red_e] (4.0, -0.5) .. controls (5.2, -0.5) and (5.4, -0.5) .. (5.4, 0);

        \draw[red_e] (6, 0.66) .. controls (6.5, 0.66) and (7.6, -3.17) .. (8,-3.163);

        \draw (10,-2) .. controls (11,-2) and (11,-5) .. (10,-5);
        \draw (10,-5) .. controls (9,-5) and (1,-5) .. (0,-5.5);
        \draw (0,-5.5) .. controls (-1,-6) and (-1, -10) .. (0, -10);

        \draw[red_e] (9.3,-2) .. controls (9.3,-1) and (11.5,-1) .. (11.5,-5);
        \draw[red_e] (11.5,-5) .. controls (11.5, -9) and (10, -9) .. (9,-9);
        \draw[red_e] (9,-9) .. controls (7,-9) and (1.12, -9) .. (1.12,-8);

        \draw (0,-10) .. controls (0.0,-9.5) .. (1.5,-9.5);
        \draw (1.5,-9.5) .. controls (2.5,-9.5) .. (2.5, -12);
        \draw (2.5,-12) .. controls (2.5,-12.5)  .. (3.5, -12.5);
        \draw (3.5,-12.5) .. controls (3.5,-12.5) and (4, -12.5)  .. (4, -12);

        \draw[red_e] (2,-7.47) .. controls (3.5,-7.47) and (2.5,-11.47) .. (4,-11.47);

        \draw[red_e] (-0.01,0.945) -- (0.82,2);
        \draw[red_e] (5.4,0) -- (6,0.66);
        \draw[red_e] (8,-3.163) -- (10,-2) -- (9.3,-2);
        \draw[red_e] (1.12,-8) -- (2,-7.47);
        \draw[red_e] (4,-11.47) -- (6.01,-10.88);
      \end{scope}
      \node at (2,-3) {a)};
      \node at (9.5,-3) {b)};
    \end{tikzpicture}
  }
  \caption{Illustration of the graph matching algorithm.
           \newline
           a) Input to the problem.
           Graph $H$ is drawn in \textcolor{blue}{blue},
           while graph $G$ in drawn in black.
           A ``man's path'' that allows the ``dog'' to visit
           all of $H$ with a leash of (optimal) length 1
           is marked in \textcolor{red}{red}.
           \newline
           b) Linked \emph{free space} diagrams for the
           decision algorithm with a leash of length 1
           form a graph (one $C_{g,h}$ diagram for each pair of edges
           $(g, h)$, $g\in G$, $h\in H$; white area represents
           \emph{free space}; edges connect respective facets
           of diagrams reachable via \emph{free space});
           for details, refer to~\cite{Alt2003262}.
           Path corresponding
           to the ``man's path'' from a)
           is marked in \textcolor{red}{red}.}
  \label{H_G_graphs}
\end{figure*}
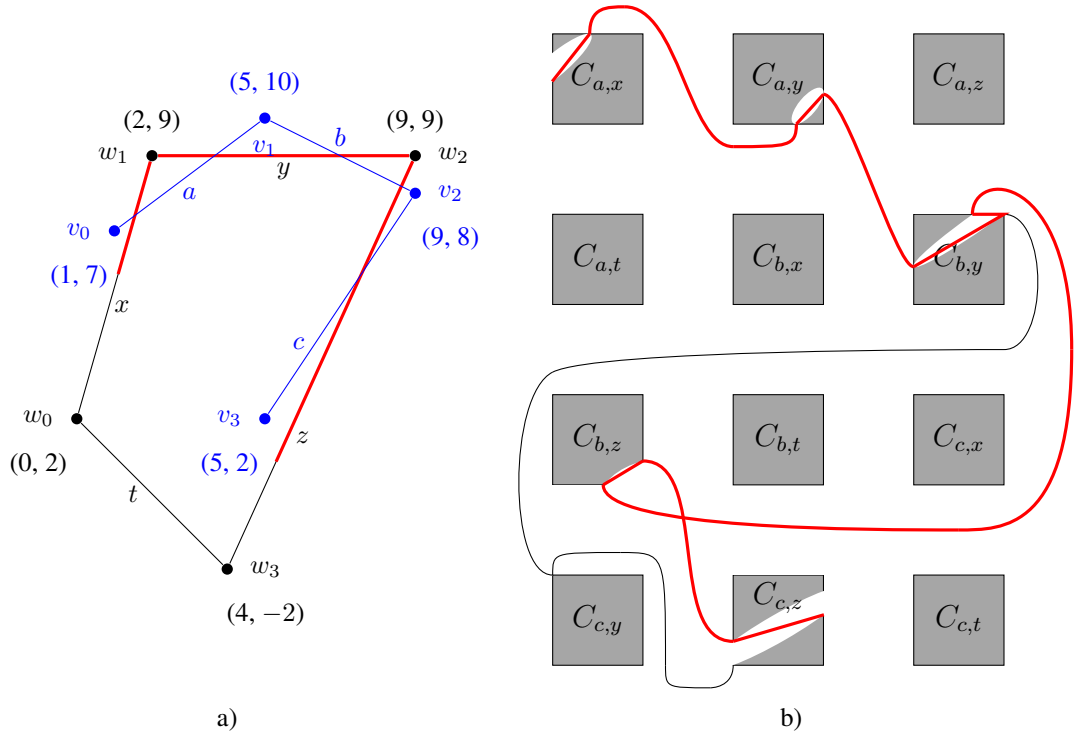 

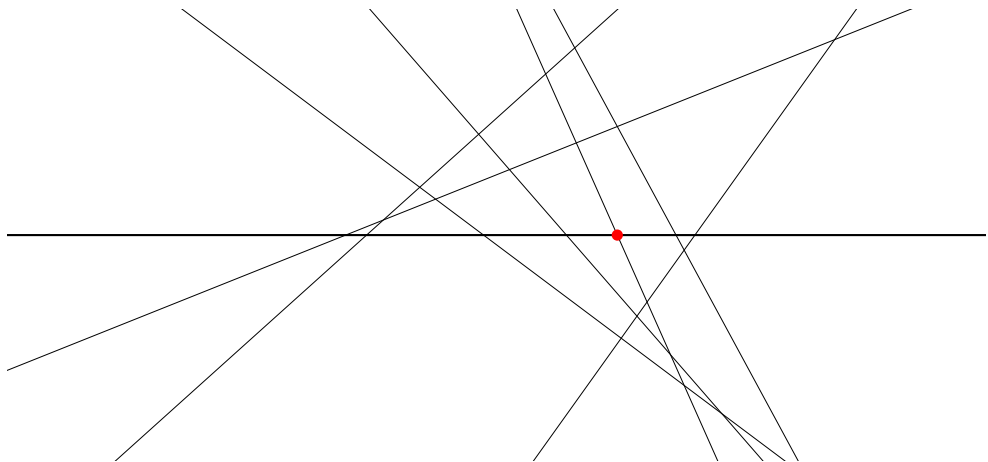
\begin{figure*}[p] 
  \begin{center}
    \begin{tikzpicture}
      [
        l/.style={very thin},
        bl/.style={thick}
      ]

      \clip (-7, -3) rectangle (6, 3);
      \draw[bl, name path=OX] (-10, 0) -- (10, 0);

      \coordinate (l1) at (-10, -3);
      \coordinate (r1) at (10, 5);

      \coordinate (l2) at (-10, -7);
      \coordinate (r2) at (10, 11);

      \coordinate (l3) at (-10, -17);
      \coordinate (r3) at (10, 11);

      \coordinate (l4) at (-10, 7);
      \coordinate (r4) at (10, -8);

      \coordinate (l5) at (-10, 12);
      \coordinate (r5) at (10, -11);

      \coordinate (l6) at (-10, 22);
      \coordinate (r6) at (10, -15);

      \coordinate (l7) at (-10, 25);
      \coordinate (r7) at (10, -20);

      \draw[l] (l1) -- (r1);
      \draw[l] (l2) -- (r2);
      \draw[l] (l3) -- (r3);
      \draw[l] (l4) -- (r4);
      \draw[l] (l5) -- (r5);
      \draw[l] (l6) -- (r6);
      \draw[l, name path=MED] (l7) -- (r7);

      \path[name intersections={of=OX and MED}];

      \node[fill=red,circle,inner sep=1.5pt] at (intersection-1) {};
    \end{tikzpicture}
  \end{center}
  \caption{Illustration of the median-of-lines problem with 7 lines.
           Median point is marked in \textcolor{red}{red}
           (there are 3 lines above, and 3 lines below it).}
  \label{median_of_lines_example}
\end{figure*} 

\subsubsection{Test Summary}
Test results reveal some interesting properties of the underlying problems.

For the point labeling problem, both algorithms
perform about 20 calls to $\mathcal{C}$, regardless of the input size.
In this case the \verb\quicksort\-based algorithm works about 15\% faster.

For the median-of-lines and graph matching problems,
the required number of calls to $\mathcal{C}$ grows steadily
as input size increases. In these cases, both algorithms run
in virtually the same time, with the \verb\boxsort\-based algorithm
seemingly gaining advantage as input size grows.

We note that the \verb\quicksort\-based algorithm might be
favored by our choice of input data, as it was
shown~\cite{DBLP:journals/comgeo/OostrumV04}
that for sufficiently random input (defined as having
close to uniform distribution of the roots of
comparison polynomials),
the algorithm requires only $O(\log n)$ calls to $\mathcal{C}$.

Our algorithm, on the other hand, always
requires only $O(\log n)$ calls to $\mathcal{C}$ with high probability,
\emph{regardless of the input}.
Based on the performed tests, we can say that our solution
is practically competitive with the \verb\quicksort\-based method,
while having superior provable bounds on the running time.

\subsection{Example Instances of Implemented Problems}
In Figures~\ref{point_labeling_example}, \ref{H_G_graphs}
and \ref{median_of_lines_example}, we show sample inputs
and results for the test problems.

\end{document}